\documentclass{article}

\usepackage{bm}
\usepackage{amsfonts}
\usepackage{amssymb}
\usepackage{amsmath}
\usepackage{verbatim}
\usepackage{color}
\usepackage{amsthm}
\usepackage{braket}
\usepackage{bbm}
\usepackage{textcomp}

\newcommand{\eins}{{\mathbbm 1}}
\newcommand{\st}{{\mathcal S}}
\newcommand{\hr}{{\mathcal H}}

\newcommand{\kr}{{\mathcal K}}
\newcommand{\bo}{{\mathcal B}}
\newcommand{\cc}{{\mathbb C}}

\newcommand{\nn}{{\mathbb N}}

\newtheorem{theorem}{Theorem}

\newtheorem{corollary}[theorem]{Corollary}
\newtheorem{definition}{Definition}
\newtheorem{example}[theorem]{Example}
\newtheorem{lemma}{Lemma}
\newtheorem{proposition}[theorem]{Proposition}
\newtheorem{remark}{Remark}

\newcommand{\tr}{\mathrm{tr}}
\newcommand{\supp}{\mathrm{supp}}

\DeclareMathOperator{\rank}{rank}

\begin{document}
\title{Universal quantum state merging} 
\author{Igor Bjelakovi\'c, Holger Boche, Gisbert Jan\ss en\\
\scriptsize{Electronic addresses: \{igor.bjelakovic, boche, gisbert.janssen\}@tum.de}
\vspace{0.2cm}\\
{\footnotesize Lehrstuhl f\"ur Theoretische Informationstechnik, Technische Universit\"at M\"unchen,}\\
{\footnotesize 80290 M\"unchen, Germany }}
\maketitle
\begin{center}
\small \emph{Dedicated to the memory of Rudolf Ahlswede}
\end{center}
\vspace{0.3cm}

\begin{abstract}
We determine the optimal entanglement rate of quantum state merging when assuming that the state is unknown except for its  membership in a certain set of states. We find that merging is possible at the lowest rate allowed by the individual states. Additionally, we establish a lower bound for the classical cost of state merging under state uncertainty. To this end we give an elementary proof for the cost in case of a perfectly known state which makes no use of the ``resource framework''. 
As applications of our main result, we determine the capacity for one-way entanglement distillation if the
source is not perfectly known. Moreover, we give another achievability proof for the entanglement generation capacity over compound quantum channels.
\end{abstract}
\maketitle

\tableofcontents

\section{\label{sect:introduction}Introduction}
Quantum state merging was introduced by Horodecki, Oppenheim, and Winter \cite{horodecki05, horodecki07b} in order 
to quantify the amount of partial quantum information contained in bipartite quantum states. I.e. for a bipartite i.i.d. 
quantum source with generic state $\rho_{AB}$ shared by communication parties  A (``sender'') and B (``receiver''), we 
want to know how much quantum communication is needed per copy when transferring A's share to B so that source output is 
completely available to B. \\
A convenient way of measuring quantum communication within this scenario is quantifying the entanglement cost 
(cf. Ref. \cite{horodecki07b}): The parties A and B are free to use local operations together with certain
exchange of classical messages (LOCC) and moreover they may use preexistent pure entanglement. The protocol performs 
state merging and produces/returns pure entanglement. The optimal rate for this task was determined in Ref. \cite{horodecki07b} 
as the conditional von Neumann entropy $S(A|B)$. In this way, the conditional von Neumann entropy obtains an operational 
interpretation as the net amount of entanglement resources needed to merge the states. Moreover, the 
puzzling fact that for some states $S(A|B)<0$ can occur can be interpreted naturally within the state merging paradigm: 
Merging protocols achieving negative rates produce rather than consume entanglement during the process.\\ 
Additionally, the optimal (i.e. the lowest possible) classical communication rate for a merging procedure achieving 
quantum rate $S(A|B)$ was determined in Ref. \cite{horodecki07b} as well. It turned out that $I(A;E)$, the quantum mutual 
information between the $A$-part and an environment purifying $\rho_{AB}$ is optimal in this case.\\
Another important aspect is that many other protocols can be derived (mostly by reduction) from quantum state merging. 
Here we just mention some of the examples from \cite{horodecki07b} like distributed compression, quantum source coding 
with side information at the decoder, and entanglement generation over quantum multiple access channels.\\
However, these results rely on the assumption of idealized conditions. The authors of Ref. \cite{horodecki07b} assumed the 
source to be memoryless and perfectly known. Both of these conditions will hardly be fulfilled in real-life communication 
settings.\\
In this paper, we drop the second condition and determine the optimal average cost of entanglement per copy under partial 
ignorance of the state to be merged. We consider a scenario, where statistical properties of the ensemble emitted by the 
source are not perfectly known to the merging partners. Rather it is assumed, that they only know that the state belongs 
to a certain set of states. Thus they have to use a protocol which works well for every member of this set. This model 
can be seen as a source analogue to the notion of compound quantum channels which were considered in Refs.\cite{bjelakovic08c} and
\cite{bjelakovic09d}.\newline
Our main technical result is a generalization of the original one-shot bound given in Ref. \cite{horodecki07b}, which respects 
state uncertainty. The question of the optimal classical communication cost in this case is addressed as well.\newline
The results of this paper gather their relevance from the fact, that other related communication protocols can be 
obtained by modifying state merging protocols. Our generalization to sets of states can be used to generate protocols 
which are successful in the corresponding ``compound'' scenarios. These in turn are stepping stones to tackle the much 
more involved ``arbitrarily varying'' models. If one considers, for example, the problem of determining capacities of 
arbitrarily varying channels, it is well known that good codes for particular compound channels can be transformed in 
good random codes for arbitrarily varying channels via Ahlswede's robustification technique \cite{ahlswede80}. The 
robustification technique can be applied in the quantum case as well. It is exactly this idea that was employed in 
Ref. \cite{ahlswede10a} in order to determine the random code capacity for entanglement transmission over arbitrarily varying 
quantum channels. This in turn can be used to show that either the deterministic classical capacity of the arbitrarily 
varying quantum channel is zero or the deterministic and random code capacities for entanglement transmission of these 
channels are equal, a quantum version of Ahlswede's famous dichotomy \cite{ahlswede78}.\\
We mention this here, because this is up to date the only method allowing us to prove such results. The ingenious 
and very direct method to prove the coding theorem for classical arbitrarily varying channels developed by 
Csisz\'ar and Narayan \cite{csiszar88} does not carry over to the quantum case.

\begin{subsection}{\label{subsect:rel_work}Related work}
The present result relies, as it was in the single state case, on a variant of the so-called decoupling approach, an idea 
which originally appeared in Ref. \cite{schumacher02} and was successfully applied to several scenarios. 
The idea is, in short, to consider not only the bipartite states to merge, but purifications of them, where the purifying 
systems are not allowed to be affected by $A$ or $B$. In this way, the question of success of the procedure is broken 
down to successful decoupling of the subsystems under control of $A$ from the purifying environment. Techniques which were
developed earlier \cite{bjelakovic08c, bjelakovic09d} for proving coding theorems for compound quantum channels
based on the decoupling approach, can be used here as well.\\
The quantum state merging protocol can be further generalized, by replacing the classical communication channels involved 
by quantum channels. This leads to the so-called fully quantum Slepian Wolf or ``mother'' protocol \cite{abeyesinghe09}, 
which together with a corresponding ``father'' protocol forms the head of a whole hierarchy of quantum protocols.
\end{subsection}

\begin{subsection}{\label{subsect:outline}Outline}
In Section \ref{sect:definitions}, we provide precise definitions for the model considered in this work. At the end of the
section, our main result is stated. Section \ref{sect:one-shot} contains the technical groundwork for the proof of our 
main result. There, we generalize the original one-shot result for single states from Ref. \cite{horodecki07b} to the case, 
where the set of possible states to merge is finite. With these results at hand, we prove our main result in Section 
\ref{merging_proof}, where we first establish the direct part in case that the set of possible states to merge is finite. 
Then we extend this result to arbitrary sets of states using finite approximations in the set of quantum states. 
The converse statement directly carries over from the known result for single states. Section \ref{sect:classical_cost} 
is devoted to the classical communication cost of quantum state merging. There we review the single state case and add an 
elementary proof to the corresponding result from Ref. \cite{horodecki07b}. Unfortunately, the protocol class used to 
establish the achievability proof for the quantum cost turns out to be too narrow. We point out, that contrary 
to the single state case, it is suboptimal regarding the classical communication requirements. \\
We conclude our work by demonstrating some applications of our main result in Sect. \ref{sect:applications}, where we determine the entanglement distillation capacity in case, that the source 
from which is distilled is not perfectly known. Finally, we give another proof for the  direct part of the entanglement 
generation coding theorem for compound quantum channels. There we use the  
correspondence between distillation of entanglement from quantum states and entanglement generation over quantum channels.
\end{subsection}

\subsection{\label{subsect:notations}Notations and Conventions}
All the Hilbert spaces which appear in this work are assumed to be finite dimensional and over the field of complex 
numbers. For any two Hilbert spaces $\hr$
and $\kr$, $\bo(\hr, \kr)$ denotes the set of linear operators mapping $\hr$ to $\kr$ and $\bo(\hr)$ denotes the set of 
linear operators on $\hr$. The set of states (i.e. positive semidefinite operators of trace one) on $\hr$ is denoted by $\st(\hr)$. 
With a Hilbert space $\kr$, the set of channels (i.e. completely positive (cp) and trace preserving maps) from $\bo(\hr)$ to 
$\bo(\kr)$ is denoted by $\mathcal{C}(\hr, \kr)$, the set of trace non-increasing cp maps by 
$\mathcal{C}^{\downarrow}(\hr,\kr)$. With a little abuse of notation, we write $id_\hr$ for the identical channel on 
$\bo(\hr)$. Because we mainly deal with systems containing several relevant subsystems, we freely make use of the 
following convention: An Hilbert space $\hr_{XYZ}$ is always thought to be the space of a composite system consisting 
of systems with Hilbert spaces $\hr_X$, $\hr_Y$ and $\hr_Z$. We use a similar notation for 
states of composite systems. A state denoted $\rho_{XY}$ for instance is a bipartite state with marginals $\rho_X$ and 
$\rho_Y$ and so on. Pure states on $\hr$ are identified with state vectors, e.g. the symbol $\psi$ sometimes denotes 
the state $\ket{\psi}\bra{\psi}$ and sometimes a state vector $\psi \in \hr$ corresponding to $\ket{\psi}\bra{\psi}$. The fidelity is defined by 
\begin{align*}
 F(\rho,\sigma) := \|\sqrt{\rho}\sqrt{\sigma}\|_1^2 
\end{align*}
for quantum states $\rho$ and $\sigma$ on a Hilbert space $\hr$. We frequently use the fact that if one of the input 
states is pure, the fidelity takes the form of an inner product
\begin{align}
F(\rho, \ket{\psi}\bra{\psi}) = \braket{\psi, \rho \psi}. \label{fidelity_pure}
\end{align}
For other properties of the fidelity see Ref. \cite{jozsa94}. The von Neumann entropy of a state $\rho$ is defined 
\begin{align*}
 S(\rho) := -\tr(\rho \log \rho) 
\end{align*}
where $\log(\cdot)$ denotes the base two logarithm throughout this work (accordingly $\exp(\cdot)$ is defined to base two as well). For certain other information quantities 
we choose a notation which indicates the states on which they are evaluated. For a state $\rho_{XY}$ on $\hr_{XY}$ we
denote the quantum mutual information by
\begin{align*}
 I(X;Y, \rho_{XY}) := S(\rho_X) + S(\rho_Y) - S(\rho_{XY}), 
\end{align*}
and the conditional von Neumann entropy by
\begin{align*}
 S(X|Y, \rho_{XY}) := S(\rho_{XY}) - S(\rho_Y).
\end{align*}
For a channel 
$\mathcal{N}\in \mathcal{C}(\hr,\kr)$ and and a state $\rho \in \st(\hr)$, the coherent information is denoted by
\begin{align*}
 I_c(\rho,\mathcal{N}) := S(\mathcal{N}(\rho)) - S((id_{\hr} \otimes \mathcal{N})(\ket{\varphi}\bra{\varphi})), 
\end{align*}
where $\varphi$ is an arbitrary purification of $\rho$ on $\hr \otimes \hr$. We further denote the hermitian conjugate of an 
operator $a$ by $a^\ast$ and the complex conjugate of a complex number $z$ by $\overline{z}$. 
We use $[N]$ as the shortcut for the set $\{1,...,N\}$ for $N \in \nn$.\\ 
Concluding this section, we specify the notion of one-way LOCC channels. As references, we recommend Ref. 
\cite{keyl02} (were the following definitions can be found stated in the Heisenberg picture), and the more 
recent treatment Ref. \cite{chitambar12}. 
Readers not familiar with LOCC channels may also consult the appendix on the same topic included in this 
paper, where the following definitions are stated more extensively.\newline
A \emph{quantum instrument} (or just \emph{instrument}) on a Hilbert space $\hr$ can be defined as a family 
$\{\mathcal{T}_k\}_{k=1}^D \subset \mathcal{C}^{\downarrow}(\hr,\kr)$ of trace non-increasing cp maps with an output space 
$\kr$ such that their sum is a channel, i.e. $\sum_{k=1}^D \mathcal{T}_k(\cdot)$ is trace preserving. We will only consider 
finite families (i.e. $D$ finite) in this paper. For bipartite Hilbert spaces $\hr_{AB}$ and $\kr_{AB}$, 
a channel $\mathcal{N} \in \mathcal{C}(\hr_{AB}, \kr_{AB})$ is called an \emph{LOCC channel with one-way classical 
communication from $A$ to $B$} (or \emph{$A \rightarrow B$ one-way LOCC} for short), if it is a combination of an instrument 
$\{\mathcal{A}_k\}_{k=1}^D \subset \mathcal{C}^{\downarrow}(\hr_A,\kr_A)$ on $A$'s systems and a family of quantum channels 
$\{\mathcal{B}_k\}_{k=1}^D$ on $B$'s systems in the following manner. To each member $\mathcal{A}_k$ of the instrument there is assigned 
a channel $\mathcal{B}_k$ resulting in the form
\begin{align}
  \mathcal{N}(\rho) = \sum_{k=1}^D \mathcal{A}_k \otimes \mathcal{B}_k(\rho) &&(\rho \in \st(\hr_{AB})).  \label{locc_allg}
\end{align}
The interpretation of (\ref{locc_allg}) is, that $B$ 
chooses a channel for his system which depends on which of the $D$ operations has been realized on $A$'s system.\\
The amount of $A \rightarrow B$ classical communication required for application of $\mathcal{M}$ is therefore determined 
by the possible measurement outcomes assigned to the operations $\mathcal{A}_1,...,\mathcal{A}_D$, i.e. a message of 
lenght $\lceil \log D \rceil$ bits has to be communicated. 
\section{\label{sect:definitions}Definitions and main result}
Let $\mathcal{X} \subseteq \st(\hr_{AB})$ be a set of bipartite states with subsystems distributed to (possibly) distant 
communication partners A and B. An $(l,k_l)$-merging for $\mathcal{X}$ is an one-way LOCC channel 
\begin{align*}
 \mathcal{M}_l: \bo(\kr_{AB}^0)\otimes \bo(\hr_{AB}^{\otimes l}) \rightarrow \bo(\kr_{AB}^1) \otimes 
 \bo(\hr_{B'B}^{\otimes l}),
\end{align*}
with local operations on the $A$- and the $B$-subscripted spaces and classical $A \rightarrow B$ communication, 
where $\kr_A^i \simeq \kr_B^i$ for $i= 0,1$ and 
$k_l := \dim\kr_A^0/\dim\kr_A^1$. A real number $R$ is called an achievable \emph{entanglement rate} for $\mathcal{X}$, 
if there exists a sequence of $(l,k_l)$-mergings with
\begin{enumerate}
 \item $\underset{l \rightarrow \infty}{\limsup} \frac{1}{l}\log(k_l) \leq R$ 
 \item $\underset{\mathcal{X}_p}{\inf} F(\mathcal{M}_l\otimes id_{\hr_E^{\otimes l}}(\phi_0^l\otimes 
	  \psi_{ABE}^{\otimes l}), \phi_1^l \otimes \psi_{B'BE}^{\otimes l}) 	\rightarrow 1$ \label{fidelity_condition}
	for $l \rightarrow \infty$.
\end{enumerate}
where $\phi_0^l \in \st(\kr_{AB}^{0,l})$ and $\phi_1^l \in \st(\kr_{AB}^{1,l})$ are maximally entangled states on their 
spaces. We demand that the Schmidt ranks of these states do not grow more than exponentially fast for $l \rightarrow 
\infty$, i.e. $\dim \kr_{A}^{0,l}, \dim\kr_{A}^{1,l} \leq 2^{lC}$ for all $l \in \nn$ and some constant $C>0$. Note that 
the fraction $\dim\kr_A^{0,l}/\dim\kr_A^{1,l}$ equals, by definition, the fraction of the Schmidt 
ranks of the input and output entanglement resources $\phi_0^l$ and $\phi_1^l$. Therefore, the expression 
$\frac{1}{l}\log(k_l)$ corresponds to the number of maximally entangled qubits (ebits) per input copy consumed 
(or gathered) by the action of $\mathcal{M}_l$.\\
The infimum in the second condition is evaluated over a set $\mathcal{X}_p$ which contains a purification $\psi_{ABE}$ 
on a space $\hr_{ABE}$ for each $\rho_{AB}$ in $\mathcal{X}$. $\psi_{B'BE}$ is the state $\psi_{ABE}$ where the A-part 
is located on a Hilbert space $\hr_{B'}$ under B's control.
The fidelity measure in \ref{fidelity_condition}.) is independent of the choice of the purifications (which will be shown 
in the next section). We frequently use  the abbreviation 
\begin{align*}
 &F_m(\rho_{AB}, \mathcal{M}) := F(\mathcal{M}\otimes id_{\hr_E}(\phi_0\otimes \psi_{ABE}),
  \phi_1 \otimes \psi_{B'BE})
\end{align*}
for a state $\rho_{AB}$ and a merging channel $\mathcal{M}$ for $\rho_{AB}$ and frequently not specify the space $\hr_E$ 
explicitly. The maximally entangled input and output 
states $\phi_0$ and $\phi_1$ are considered to be determined by $\mathcal{M}$. The optimal entanglement rate 
$C_m(\mathcal{X})$, i.e. 
\begin{align*}
 C_m(\mathcal{X}) := \inf\{R : R\; \text{is an achievable entanglement rate for}\, \mathcal{X}\} 
\end{align*}
is called the \emph{merging cost} of $\mathcal{X}$.\\
The main result of this paper is the following theorem, which quantifies the merging cost of any set $\mathcal{X}$ of bipartite states.
\begin{theorem}\label{merging_theorem}
 Let $\mathcal{X} \subset \st(\hr_{AB})$ be a set of states on $\hr_{AB}$. Then 
\begin{align}
 C_m(\mathcal{X}) = \sup_{\rho \in \mathcal{X}} S(A|B;\rho) \label{theorem}
\end{align}
holds.
\end{theorem}
To prove the achievability part of the above Theorem \ref{merging_theorem} we show that we find universal protocols for 
state merging within the class of LOCC operations which was used by the authors of Ref. \cite{horodecki07b}. We give a brief 
outline of our proof of Theorem \ref{merging_theorem}. In Sect. \ref{merging_fidelity_props} we state and prove some 
important facts about the fidelity measure under consideration.  We follow this path and recall the decoupling lemma 
given in Ref. \cite{horodecki07b} in Sect. \ref{subsect:protocol}. On this basis we establish a one-shot bound for finite sets 
of states in Section \ref{one_shot_finite}. To this end we utilize techniques developed in Refs. \cite{bjelakovic08c} and 
\cite{bjelakovic09d} for proving coding theorems for compound quantum channels. In Sect. \ref{subsect:finite} we provide 
the direct part of our merging theorem for finite sets of states and extend these results to arbitrary sets in Sect. 
\ref{subsect:general}. The converse theorem easily carries over from the one given in Ref. \cite{horodecki07b}, and we just provide the missing 
link in Sect. \ref{subsect:converse}. \\
\section{\label{sect:one-shot}One-shot result}
\subsection{\label{merging_fidelity_props}Properties of the fidelity measure}
In this section we aim to prove some important properties of the merging fidelity. 
\begin{lemma}\label{merg_fid_rep}
Let $\mathcal{M}: \bo(\kr_{AB}^{0} \otimes \hr_{AB}) \rightarrow \bo(\kr_{AB}^{1} \otimes \hr_{B'B})$ be a channel,
$\phi_0 \in \st(\kr_{AB}^0)$, and $\phi_1 \in \st(\kr_{AB}^1)$ maximally entangled states. Then the following assertions 
hold
\begin{enumerate}
\item For any state $\rho_{AB} \in \st(\hr_{AB})$ on $\hr_{AB}$ with purification $\psi_{ABE} \in \st(\hr_{ABE})$,
  \begin{align*}
  F(\mathcal{M}\otimes id_{\hr_E}(\phi_0 \otimes \psi_{ABE}), \phi_1 \otimes \psi_{B'BE}) 
	= \sum_{z=1}^Z |\tr(p_z\rho_{AB})|^2
	  \end{align*}
    holds, where $p_1,...p_Z$ are elements of $\bo(\hr_{AB})$ which depend on $\mathcal{M}, \phi_0$ and $\phi_1$.
\item
Merging fidelity is a convex function of the input state. For any two states $\rho_1$
 and $\rho_2$ on $\hr_{AB}$ and $\lambda \in [0,1]$
 \begin{align*}
  F_m(\lambda \rho_1 + (1 - \lambda) \rho_2, \mathcal{M}) \; \leq  \; \lambda F_m(\rho_1, \mathcal{M}) 
 + (1 - \lambda)F_m(\rho_2, \mathcal{M}) 
 \end{align*}
 holds.
\end{enumerate}
\end{lemma}
\begin{proof}
Let 
\begin{align*}
  \mathcal{M}(\cdot) := \sum_{z=1}^Z m_z(\cdot)m_z^\ast 
\end{align*}
be a Kraus decomposition of $\mathcal{M}$ with operators $m_z \in \bo(\kr_{AB}^0 \otimes \hr_{AB}, \kr_{AB}^1 \otimes 
\hr_{B'B})$ for every $z \in \{1,...,Z\}$. We define channels $\mathcal{V}$ and $\mathcal{W}$ which incorporate the 
input and output states $\phi_0$ and $\phi_1$. Let $\mathcal{V} \in \mathcal{C}(\hr_{AB}, \kr_{AB}^1 \otimes \hr_{B'B})$ 
be the channel constituted by Kraus operators $v_z \in \bo(\hr_{AB},\kr_{AB}^1 \otimes \hr_{B'B})$ defined by 
\begin{align*}
v_z x := m_z(\phi_0 \otimes x)
\end{align*}
for every $1 \leq z \leq Z$, $x \in \hr_{AB}$ and $\mathcal{W}(\cdot) := w(\cdot)w^\ast$ with 
\begin{align*}
w x := \phi_1 \otimes (U \otimes \eins_{\hr_B})x 
\end{align*}
for every $x \in \hr_{AB}$. Here, $U \in \bo(\hr_A,\hr_{B'})$ is the isometry which identifies $\hr_A$ and $\hr_{B'}$. 
With these definitions at hand we have
\begin{align}
 &  F(\mathcal{M}\otimes id_{\hr_E}(\phi_0 \otimes \psi_{ABE}), \phi_1 \otimes \psi_{B'BE})   \nonumber \\
 =& F(\mathcal{V} \otimes id_{\hr_E}(\psi_{ABE}), \mathcal{W} \otimes id_{\hr_E}(\psi_{ABE})) \nonumber \\
 =& \sum_{z=1}^Z \braket{(w \otimes \eins_{\hr_E}) \psi, 
      (v_z \otimes \eins_{\hr_E}) \ket{\psi}\bra{\psi} (v_z \otimes \eins_{\hr_E})^\ast (w \otimes \eins_{\hr_E}) \psi}. 
    \label{merg_fid_sk} 
\end{align}
The r.h.s. of (\ref{merg_fid_sk}) is due to the fact that the fidelity admits a representation in terms of an inner 
product if one of the inputs is pure, see eq. (\ref{fidelity_pure}). 
Each of the summands on the r.h.s. of eq. (\ref{merg_fid_sk}) can be written as
\begin{align}
\braket{\psi, (w^\ast v_z \otimes \eins_{\hr_E}) \ket{\psi}\bra{\psi} (v_z^\ast w \otimes \eins_{\hr_E}) \psi}
&= \braket{\psi, (w^\ast v_z \otimes \eins_{\hr_E}) \psi}\overline{\braket{\psi, (w^\ast v_z \otimes \eins_{\hr_E})\psi}} 
  \nonumber\\
&= |\tr((w^\ast v_z\otimes \eins_{\hr_E}) \ket{\psi}\bra{\psi})|^2 \nonumber\\
&= |tr(w^\ast v_z \rho_{AB})|^2. \label{merg_fid_sk_2}
\end{align}
Inserting the r.h.s. of eq. (\ref{merg_fid_sk_2}) into (\ref{merg_fid_sk}) yields
\begin{align*}
F(\mathcal{M}\otimes id_{\hr_E}(\phi_0 \otimes \psi_{ABE}), \phi_1 \otimes \psi_{B'BE}) = 
 \sum_{z=1}^Z|\tr(w^\ast v_z \rho_{AB})|^2,
\end{align*}
which is the desired result, if we set $p_z = w^\ast v_z$ for every $z$. The second assertion of the lemma is a direct consequence of the first one together with the fact that the fidelity takes only values in $[0,1]$.
\end{proof}

\subsection{\label{subsect:protocol}Protocol and decoupling for single states}
In this section we briefly recall a result given in Ref. \cite{horodecki07b} which marks the starting point for our investigations. 
Fortunately, the protocol constructed there, which is of relatively simple structure, can be modified for our purposes. 
Let $d_A$ be the dimension of the Hilbert space $\hr_A$. For an integer $0 < L \leq d_A$ we use the term 
$L$-\emph{merging} if we speak of a channel
\begin{align*}
 \mathcal{M}: \bo(\hr_{AB}) \rightarrow \bo(\kr_{AB}) \otimes \bo(\hr_{B'B}) 
\end{align*}
which is of the form
\begin{align}
 \mathcal{M}(\rho) = \sum_{k=0}^D a_k \otimes u_k (\rho) a_k^\ast \otimes u_k^\ast,  \label{lmerging}
\end{align}
for every $\rho \in \st(\hr_{AB})$. Here $D$ is defined $D:=\lfloor \frac{d_A}{L}\rfloor$ and $\kr_A$ and $\kr_B$ are 
Hilbert spaces with $\dim \kr_A = \dim \kr_B = L$ and $\kr_A \subseteq \hr_A$ is a subspace of $\hr_A$, where
\begin{itemize}
 \item $\{a_k\}_{k=0}^D \subset \bo(\hr_A,\kr_{A})$ is a set of rank $L$ partial isometries (except $a_0$ which has rank 
        $d_A - L\cdot 	D < L$) with pairwise orthogonal initial subspaces (in the following, we call such channels 
	\emph{$L$-instrument} for short).
\item $\{u_k\}_{k=0}^D \subset \bo(\hr_B, \kr_B \otimes \hr_{B'B})$ is a family of isometries.
\end{itemize}
We abbreviate the corresponding operation with $\mathcal{A}_k := a_k(\cdot)a_k^\ast$ for every $k$. 
Let $\psi_{ABE}$ be a purification of $\rho_{AB}$ on a Hilbert space $\hr_{ABE}$. For notational simplicity we define 
abbreviations
\begin{align*}
 p_k := \tr(a_k\rho_A a_k^\ast)\hspace{0.3cm} \text{and} \hspace{2.0cm} 
 \rho_{AE}^k := \tr_{\hr_B}(\mathcal{A}_k\otimes id_{\hr_E}(\psi_{ABE})).
\end{align*}
for every $k \in \{0,...,D\}$.
The following lemma is taken from Ref. \cite{horodecki07b}, we repeat it here including a sketch of the proof which we give for the convenience of the reader.
\begin{lemma}[cf. Ref. \cite{horodecki07b}, Prop. 3]\label{lemma:one-shot-single}
Let $\rho_{AB}$ be a bipartite state on $\hr_{AB}$ and $\{a_k\}_{k=0}^D \subset \bo(\hr_A, \kr_A)$ an $L$-instrument. 
There exists a family $\{u_k\}_{k=0}^D$ of isometries completing $\{a_k\}_{k=0}^D$ to an $L$-\emph{merging} $\mathcal{M}$
which satisfies
\begin{align*}
  F(\mathcal{M}\otimes id_{\hr_E}(\psi_{ABE}), \phi_L \otimes \psi_{B'BE})\geq 1- \tilde{Q},
\end{align*}
where $\tilde{Q}$ is defined by
\begin{align}
  \tilde{Q} := 2 \left( p_0 + \sum_{k=1}^D \left\|\rho_{AE}^k - \frac{L}{d_A} \pi_{L} \otimes \rho_E \right\|_1 \right).
\label{lemma:one-shot-single_1}
\end{align}
Here, the state $\phi_L$ is maximally entangled on $\kr_{AB}$ and $\pi_{L}$ denotes the maximally mixed state on $\kr_A$ 
(i.e. $\pi_L := \frac{\mathbbm{1}_{\kr_A}}{L})$.
\end{lemma}
In the following proof, the well known relations (see Ref. \cite{Fuchs2007a})
\begin{align}
F(\rho, \sigma) &\geq 1 -\|\rho - \sigma \|_1 \hspace{0.3cm} \text{and} \label{fvg_1} \\
\|\rho -  \sigma \|_1 &\leq 2 \sqrt{1-F(\rho,\sigma)} \label{fvg_2}
\end{align}
between trace distance and fidelity of any two states $\rho$ and $\sigma$ on a Hilbert space $\hr$ are used.
\begin{proof}
For every $k$, $0 \leq k \leq D$, the (sub-normalized) state $\mathcal{A}_k \otimes id_{\hr_E}(\psi_{ABE})$ is a 
purification of $\rho_{AE}^k$ and $\phi_L \otimes \psi_{B'BE}$ is a purification of $\pi_L \otimes \rho_E$. These facts 
and Uhlmann's theorem \cite{uhlmann76} (see Ref. \cite{jozsa94} for the finite dimensional version) guarantee that for every 
$k \in \{0,...,D\}$ there exists an isometry $u_k : \hr_B \rightarrow \kr_B \otimes \hr_{B'B}$ satisfying
\begin{align}
    F(\mathcal{A}_k \otimes \mathcal{U}_k \otimes id_{\hr_E}(\psi_{ABE}), \phi_L \otimes \psi_{B'BE}) 
  = F(\rho_{AE}^k,\pi_L \otimes \rho_E), \label{uhlmann_konsequenz}
\end{align}
where $\mathcal{U}_k(\cdot) := u_k(\cdot)u_k^\ast$. The rest is mostly done by lower bounding the fidelity in terms of 
the trace distance. Given the case that $p_k > 0$ for $k$, using (\ref{fvg_1}) we have 
\begin{align}
 F(\rho_{AE}^k,\pi_L \otimes \rho_E)  &= p_k F\left(\frac{1}{p_k}\rho_{AE}^k, \pi_L \otimes \rho_E\right) \nonumber \\
			      &\geq p_k - \|\rho_{AE}^k - p_k \pi_L \otimes  \rho_E\|_1. \label{spurabsch_uhlmann}
\end{align}
In case that $p_k = 0$ for $k$, $F(\rho_{AE}^k,\pi_L \otimes \rho_E) = 0$. Taking the sum over all $k$ we arrive 
at
\begin{align}
 F\left(\sum_{k=0}^D \mathcal{A}_k \otimes \mathcal{U}_k \otimes id_{\hr_E}(\psi_{ABE}), \phi_L \otimes \psi_{B'BE}\right)
    &= \sum_{k=0}^D F(\rho_{AE}^k, \pi_L \otimes \rho_E) \label{uhlmann_absch_1}\\
    &\geq 1 - \sum_{k=0}^D \|\rho_{AE}^k - p_k \pi_L \otimes \rho_E \|_1 \label{uhlmann_absch_2} \\
    &\geq 1 - 2 p_0 - \sum_{k=1}^D \|\rho_{AE}^k - p_k \pi_L \otimes \rho_E \|_1. \label{uhlmann_absch_3}
\end{align}
Eq. (\ref{uhlmann_absch_1}) follows from the linearity of the fidelity in one of the inputs given the other one is pure 
and (\ref{uhlmann_konsequenz}). For (\ref{uhlmann_absch_2}) we used (\ref{spurabsch_uhlmann}) along with the fact that 
$\sum_{k=0}^D \mathcal{A}_k$ is a channel implying $\sum_k p_k = 1$. The r.h.s. of (\ref{uhlmann_absch_3}) holds because the 
trace distance of any two states is upper bounded by $2$ which ensures
\begin{align*}
 \|\rho_{AE}^0 - p_0 \pi_L \otimes \rho_E\|_1 \leq 2 p_0. 
\end{align*}
It remains to show that 
$\|\rho_{AE}^k - p_k \pi_L \otimes \rho_E \|_1 \leq \, 2\cdot\|\rho_{AE}^k - \frac{L}{d_A} \pi_L \otimes \rho_E \|_1$, which 
can be seen as follows. It holds that
\begin{align*}
  \|\rho_{AE}^k - p_k \pi_L \otimes \rho_E \|_1  
  \leq & \, \|\rho_{AE}^k - \frac{L}{d_A} \pi_L \otimes \rho_E \|_1 + |p_k - \frac{L}{d_A}| \\
  \leq & \, 2\cdot \|\rho_{AE}^k - \frac{L}{d_A} \pi_L \otimes \rho_E \|_1, 
\end{align*}
where the first inequality is obtained by adding a zero and applying the triangle inequality together with the fact that 
every quantum state has trace norm one. The second line is by monotonicity of the trace norm under the action of 
channels.
\end{proof}

\subsection{\label{one_shot_finite}One shot bound for finite sets of states}
In this section we consider a finite set $\mathcal{X} := \{\rho_{AB,i}\}_{i=1}^N$ of states on $\hr_{AB}$ and derive a 
bound for the minimal merging fidelity of the states in $\mathcal{X}$ which is based on Lemma \ref{lemma:one-shot-single}. 
The main ingredient for the proof is the observation, that a good merging scheme for the averaged state
\begin{align}
\overline{\rho}_{AB} := \frac{1}{N}\sum_{i=1}^N \rho_{AB,i} \label{average_st}
\end{align}
will be good for every single member of $\mathcal{X}$. This is due to convexity of the merging fidelity 
(see Lemma \ref{merg_fid_rep}). Now let 
$\psi_{ABE,i}$ be any purification of $\rho_{AB,i}$ on $\hr_{ABE}$ for every $i \in [N]$. The state
\begin{align}
 \ket{\overline{\psi}_{ABR}}\bra{\overline{\psi}_{ABR}} 
  := \frac{1}{N} \sum_{i,j=1}^N \ket{\psi_{ABE,i}}\bra{\psi_{ABE,j}} \otimes \ket{e_i}\bra{e_j}
\label{average_purif}
\end{align}
with $\{e_i\}_{i=1}^N$ being an orthonormal basis in $\cc^N$  is a purification of 
$\overline{\rho}_{AB}$ on $\hr_{ABR}$ with $\hr_R := \hr_E \otimes \cc^N$. The following lemma provides a lower bound 
for the fidelity of an $L$-merging of $\overline{\rho}_{AB}$ in terms of quantities determined by the states in 
$\mathcal{X}$. 
\begin{lemma} \label{lemma:pre_expectation}
Let $\{\rho_{AB,i}\}_{i=1}^N$ be a set of states on $\hr_{AB}$. Then for the corresponding averaged state
$\overline{\rho}_{AB}$ and purifications $\psi_{ABE,1},...,\psi_{ABE,N}$, Lemma 
\ref{lemma:one-shot-single} also holds with $\tilde{Q}$ replaced by 
\begin{align*}
Q := 2 \left(p_0 + \frac{1}{N} \sum_{k=1}^D
\sum_{i,j=1}^N \sqrt{L_{ij} \cdot T_{ij}^{(k)}} \right)
\end{align*}
where $L_{ij} := L \cdot \min_{m \in \{i,j\}}\{\rank(\rho_{E,m})\}$ and 
\begin{align*}
&T_{ij}^{(k)} := \left\|\rho_{AE,ij}^k - \tfrac{L}{d_A} \pi_{L} \otimes \rho_{E,ij}\right\|_2^2.
\end{align*}
Here we used the definitions 
\begin{align*}
\psi_{ABE,ij}&:= \ket{\psi_{ABE,i}} \bra{\psi_{ABE,j}},\hspace{0,1cm} \rho_{E,ij} := \tr_{\hr_{AB}}(\psi_{ABE,ij}), 
\hspace{0.1cm} \text{and} \\ 
\rho_{AE,ij}^k &:= \tr_{\hr_B}((a_k\otimes \eins_{\hr_{BE}})\psi_{ABE,ij}(a_k^{\ast}\otimes \eins_{\hr_{BE}})) 
\end{align*}
for $i,j \in [N], k \in [D]$.
\end{lemma}

\begin{proof}
Define 
\begin{align*}
  \overline{\rho}_R  := \tr_{\hr_{AB}}(\overline{\psi}_{ABR}),\hspace{0.2cm}\text{and}
  \hspace{0.2cm}\overline{\rho}_{AR}^k := 
  \tr_{\hr_B}((a_k\otimes \eins_{\hr_{BR}})\overline{\psi}_{ABR} (a_k^\ast \otimes \eins_{\hr_{BR}})) 
\end{align*}
for every $k \in [D]$. We bound the trace distance terms on the r.h.s. of (\ref{lemma:one-shot-single_1}) 
for $\rho_{AB}$ with its purification introduced in eq. (\ref{average_purif}). Explicitly, for every 
$k \in [D]$, we have 
\begin{align*}
 &\left\|\overline{\rho}_{AR}^k - \tfrac{L}{d_A} \pi_L \otimes \overline{\rho}_R \right\|_1  \\
 \overset{\text{(a)}}{=} &\left\| \frac{1}{N}\sum_{i,j=1}^N \left\{\rho_{AE,ij}^k  - 
    \frac{L}{d_A} \pi_L \otimes \rho_{E,ij}\right\} \otimes \ket{e_i}\bra{e_j}\right\|_1 \\
 \overset{\text{(b)}}{\leq} &\frac{1}{N} \sum_{i,j=1}^N \left\|\left\{\rho_{AE,ij}^k - 
  \frac{L}{d_A} \pi_L \otimes \rho_{E,ij}\right\} \otimes
  \ket{e_i} \bra{e_j}\right\|_1 \\
 \overset{\text{(c)}}{\leq} &\frac{1}{N} \sum_{i,j=1}^N \left\|\rho_{AE,ij}^k - 
      \frac{L}{d_A} \pi_L \otimes \rho_{E,ij} \right\|_1 \\
  \overset{\text{(d)}}{\leq} &\frac{1}{N} \sum_{i,j=1}^N \sqrt{L_{ij}} \left\|\rho_{AE,ij}^k - 
  \tfrac{L}{d_A} \pi_L \otimes \rho_{E,ij} \right\|_2.
\end{align*}
where $L_{ij} := L \cdot \min\{\rank(\rho_{E,i}), \rank(\rho_{E,j})\}$ for every $1 \leq i,j \leq N$. The above 
(in)equalities are justified by the following arguments. (a) by definition of 
$\overline{\rho}_R$ and $\overline{\rho}_{AR}^k$, (b) by use of the triangle inequality and (c) because the trace 
norm is multiplicative with respect to tensor products and the equality $\|\ket{e_i}\bra{e_j}\|_1 = 1$ for all 
$1 \leq i,j\leq N$. The well known relation $\|x\|_1 \leq \sqrt{r}\|x\|_2$ between the trace- and Hilbert-Schmidt 
norms with $r$ being the rank of $x$ justifies (d), if the rank of the matrix 
\begin{align*}
 \rho_{AE,ij}^k - \frac{L}{d_A} \pi_L \otimes \rho_{E,ij}
\end{align*}
is smaller or equal than $L_{ij}$ for all $i,j \in [N]$. This is fulfilled, which can be seen as follows. Let 
with an orthonormal basis $\{f_k\}_{k=1}^{\dim \hr_E}$ of $\hr_E$,
\begin{align}
 \psi_{ABE,i} := \sum_{k=1}^{r_i} \psi_{AB,k}^{(i)} \otimes f_k \label{schmidt_coeff_blmatrix}
\end{align}
be a Schmidt decomposition of $\psi_{ABE,i}$ for every $1 \leq i \leq N$, with the Schmidt coefficients incorporated
in the first tensor factors. This is always possible since we are free in the choice of the purifications. 
Using (\ref{schmidt_coeff_blmatrix}), one can verify, that
\begin{align*}
 &\rho_{AE,ij}^k - \frac{L}{d_A}\pi_L \otimes \rho_{E,ij} \\
 &= \sum_{k=1}^{r_i}\sum_{l=1}^{r_j} \left( a_k \tr_{\hr_B}(\ket{\psi_{AB,k}^{(i)}}\bra{\psi_{AB,l}^{(j)}})a_k^{\ast} 
     - \braket{\psi_{AB,l}^{(j)}, \psi_{AB,k}^{(i)}}\frac{L}{d_A}\pi_L \right) \otimes  \ket{f_k}\bra{f_l}.
\end{align*}
holds for every $i,j \in [N]$. This expression can be interpreted as an $r_i \times r_j$ block matrix with each 
block an $L \times L$ matrix. It has
therefore rank smaller or equal $L \cdot \min\{r_i, r_j\}$.
\end{proof}

Let $L \in \{1,...,d_A\}$ be fixed and an arbitrary but fixed $L$-instrument $\mathcal{A} := \{\mathcal{A}_k\}_{k=1}^D 
\subset \mathcal{C}^{\downarrow}(\hr_A, \kr_A)$ be given. Every unitary $v \in \mathfrak{U}(\hr_A)$ defines a channel 
$\mathcal{V} \in \mathcal{C}(\hr_A)$ via $\mathcal{V}(\cdot) := v(\cdot)v^\ast$. With these definitions, for every $v$, 
we get an $L$-instrument $\mathcal{A}(v)$ with
\begin{align}
 \mathcal{A}(v) := \{\mathcal{A}_k \circ \mathcal{\mathcal{V}}\}_{k=0}^D. \nonumber
\end{align}
Every collection of isometric channels $\{\mathcal{U}_k\}_{k=0}^D\subset \mathcal{C}(\hr_B, \kr_B \otimes \hr_{B'B})$ 
completes $\mathcal{A}(v)$ to an $L$-merging 
\begin{align}
 \sum_{k=0}^D \mathcal{A}_k \circ \mathcal{V} \otimes \mathcal{U}_k(\cdot)
\end{align}
We define the function
\begin{align}
F_m(\rho,\mathcal{A}(v)) := \max_{ \{\mathcal{U}_k\}_{k=0}^D} F_m(\rho_{AB}, \sum_{k=1}^D \mathcal{A}_k\circ \mathcal{V} 
			      \otimes \mathcal{U}_k(\cdot)) \label{max_fidelity_defined}
\end{align}
for every $v \in \mathfrak{U}(\hr_A)$, $\rho \in \st(\hr_{AB})$. The maximization in (\ref{max_fidelity_defined}) is over 
all collections $\{\mathcal{U}_k\}_{k=0}^D \subset \mathcal{C}(\hr_B, \kr_B \otimes \hr_{B'B})$ of isometric channels.

The expected merging fidelity under random selection of such $L$-mergings according to the normalized Haar measure on
$\mathfrak{U}(\hr_A)$ is bounded in the following lemma, which is the key technical result for the proof of the merging 
theorem.

\begin{lemma}\label{lemma:expectation}
For $L \in \{1,...,d_A\}$, a set $\{\rho_{AB,i}\}_{i=1}^N$ of states on $\hr_{AB}$ and $\psi_{ABE,i}$ a purification of 
$\rho_{AB,i}$ on $\hr_{ABE}$ for each $i$, we have
\begin{align}
  \int_{\mathfrak{U}(\hr_A)} F_m&(\overline{\rho}_{AB}, \mathcal{A}(v)) \ dv
  \geq 1 - 2 \left(\frac{L}{d_A}+ 2\cdot \sum_{i=1}^N \sqrt{L\cdot \rank(\rho_{E,i}) \| \rho_{B,i}\|_2^2} \right) 
  \label{lemma:expectation_1}
\end{align}
where the integration is with respect to the normalized Haar measure on $\mathfrak{U}(\hr_A)$.
\end{lemma}

To prove the claim of Lemma \ref{lemma:expectation} the following two lemmas are needed. 
\begin{lemma}[Ref. \cite{bjelakovic08c}, Lemma 3.2] \label{lemma:compound_lemma}
Let $L$ and $D$ be $N\times N$-matrices with nonnegative entries such that 
\begin{align} 
L_{jl} \leq L_{jj}, \hspace{0.1cm} L_{jl} \leq L_{ll}\; \text{and} \; D_{ij} \leq \max\{D_{ii}, D_{jj}\} \nonumber
\end{align}
for all $i,j \in \{1,...,N\}$. Then 
\begin{align}
 \sum_{i,j=1}^N\frac{1}{N}\sqrt{L_{ij}D_{ij}} \leq 2 \sum_{i=1}^N \sqrt{L_{ii}D_{ii}} \nonumber
\end{align}
\end{lemma}

\begin{lemma} \label{lemma:maximum_trace}
 Let $\tau$ and $\xi$ be elements of a bipartite Hilbert space $\hr \otimes \hr'$. Then
\begin{align*}
 \|\tr_{\hr'}(\ket{\tau}\bra{\xi})\|_2^2 \leq \max_{\chi \in \{\tau , \xi \}} \|\tr_{\hr'}(\ket{\chi}\bra{\chi}) \|_2^2 
\end{align*}
\end{lemma}
\begin{proof}[Proof of Lemma \ref{lemma:maximum_trace}]
Choose an orthonormal basis $\{e_m\}_{m=1}^d$ in $\hr'$ where $d:= \dim(\hr')$. The elements $\varphi$ and $\psi$ can be 
decomposed in the form
\begin{align*}
 \varphi &=  \sum_{m=1}^d \varphi_m \otimes e_m \hspace{1.0cm} \text{and}  \\
 \psi    &=  \sum_{m=1}^d \psi_m \otimes e_m 
\end{align*}
with suitable elements $\varphi_1,...,\varphi_d$ and $\psi_1,...,\psi_d$ in $\hr$. With these decompositions
\begin{align*}
 \tr_{\hr'}(\ket{\varphi}\bra{\psi}) = \sum_{m,n=1}^d \ket{\varphi_m}\bra{\psi_n}\cdot \tr(\ket{e_m}\bra{e_n}).
\end{align*}
Therefore
\begin{align}
 \|\tr_{\hr'}(\ket{\varphi}\bra{\psi})\|_2^2 &= \|\sum_{m=1}^d \ket{\varphi_m}\bra{\psi_m}\|_2^2 \label{matrixlemma_1} \\
				     &= |\tr\left(\sum_{m,n=1}^d (\ket{\varphi_m}\bra{\psi_m})^\ast(\ket{\varphi_n}
					\bra{\psi_n})\right)| \label{matrixlemma_2}\\
				     &= |\sum_{m,n=1}^d \braket{\varphi_m,\varphi_n}\braket{\psi_n,\psi_m}| 
					\label{matrixlemma_3}.
\end{align}
To show the assertion of the lemma consider 2 $d \times d$ matrices $X$ and $Y$ with entries
$X_{mn} := \braket{\varphi_m, \varphi_n}$ resp. $Y_{mn} := \overline{\braket{\psi_m, \psi_n}}$ for $0 < m,n \leq d$. Then
the r.h.s. of (\ref{matrixlemma_3}) can be read as $\tr(XY)$, and we have 
\begin{align}
 \left|\sum_{m,n=1}^d \braket{\varphi_m \varphi_n}\braket{\psi_n,\psi_m}\right| &= |\tr(XY)| \nonumber\\
								     & \leq \|X\|_2 \|Y\|_2 \label{cauchy_matrix} \\
								     & \leq \max_{Z\in \{X,Y\}}\|Z\|_2^2, \nonumber
\end{align}
where the r.h.s. of (\ref{cauchy_matrix}) is an application of the Cauchy-Schwarz inequality. It is easy to see that
$\|X\|_2^2 = \|\tr_{\hr'}(\ket{\varphi}\bra{\varphi})\|_2^2$ and $\|Y\|_2^2 = \|\tr_{\hr'}(\ket{\psi}\bra{\psi})\|_2^2$, 
so we are done.
\end{proof}

\begin{proof}[ Proof of Lemma \ref{lemma:expectation}] 
First we have to convince ourselves, that $F_m(\overline{\rho}_{AB}, \mathcal{A}(\cdot))$ depends measurably on $v \in 
\mathfrak{U}(\hr_A)$.  For each fixed set $\{\mathcal{U}_k\}_{k=0}^D$, the function 
$F_m(\overline{\rho}_{AB}, \sum_{k=1}^D \mathcal{A}_k \circ \mathcal{V} \otimes \mathcal{U}_k)$ clearly is continuous 
in $v$, therefore, $F_m(\overline{\rho}_{AB}, \mathcal{A}(v))$ as a maximum over such functions is lower semicontinous, 
which implies its measurability. \\
Using Lemma \ref{lemma:pre_expectation} we get 
\begin{align}
 F_m(\overline{\rho}_{AB}, \mathcal{A}(v)) \geq 1 - Q_v  \nonumber
\end{align}
with error
\begin{align}
 Q_v:= 2 \left(p_0^v + \frac{1}{N} \sum_{k=1}^D
\sum_{i,j=1}^N \sqrt{L_{ij} \cdot T_{ij,v}^{(k)}} \right). \nonumber
\end{align}
Here $p_0^v := \tr((\mathcal{A}_0\circ \mathcal{V})(\rho_A))$,
\begin{align*}
T_{ij,v}^{(k)} := \left\|\rho_{AE,ij,v}^k - \tfrac{L}{d_A} \pi_{L} \otimes \rho_{E,ij}\right\|_2^2
\end{align*}
and 
\begin{align*}
&\rho_{AE,ij,v}^k := \mathcal{A}_k\circ \mathcal{V}(\tr_{\hr_B}(\psi_{ABE,ij})).
\end{align*}
By virtue of Jensen's inequality
\begin{align}
 \int_{\mathfrak{U}(\hr_A)} Q_v \ dv
 \leq 2 \left(\int_{\mathfrak{U}(\hr_A)} p_0^v \ dv + \frac{1}{N} \sum_{k=1}^D \sum_{i,j=1}^N \left( L_{ij} \cdot 
  \int_{\mathfrak{U}(\hr_A)} T_{ij,v}^{(k)} \ dv \right)^{\frac{1}{2}}\right)
\nonumber
\end{align}
holds. It remains to bound the expectations in the right hand side of the above inequality. This was already done in
Lemma 6 of Ref. \cite{horodecki07b}. We have
\begin{align}
 \int_{\mathfrak{U}(\hr_A)} T_{ij,v}^k \ dv &\leq \frac{L^2}{d_A^2}\|\tr_{\hr_B}(\ket{\psi_{ABE,i}}
	  \bra{\psi_{ABE,j}})\|_2^2 , \hspace{0.2cm}\text{and} \hspace{0.4cm}  
    \int_{\mathfrak{U}(\hr_A)} p_0^v  \ dv \leq \frac{L}{d_A}.\label{haar_integrals}
\end{align}
Abbreviating $D_{ij}:= \|\tr_{\hr_B}(\ket{\psi_{ABE,i}}\bra{\psi_{ABE,j}})\|_2^2$ for every $i, j \in [N]$, 
(\ref{haar_integrals}) implies
\begin{align}
 \int_{\mathfrak{U}(\hr_A)} Q_v \ dv
&\leq 2 \left( \frac{L}{d_A} + \frac{1}{N} \sum_{k=1}^D \sum_{i,j=1}^N \sqrt{L_{ij} \cdot
\frac{L^2}{d_A^2}D_{ij}}\right)  \label{fidelitydu_0} \\
&\leq 2 \left( \frac{L}{d_A} + \frac{1}{N} \sum_{i,j=1}^N \sqrt{L_{ij} D_{ij}}\right). \label{fidelitydu}
\end{align}
The second inequality follows from the fact that the summands on the r.h.s. of (\ref{fidelitydu_0}) are independent of 
$k$ and $D\frac{L}{d_A} \leq 1$ by 
definition of $D$. By definition of $L_{ij}$, clearly $L_{ij} = \min\{L_{ii}, L_{jj}\}$ for all $i, j$ and 
so the first assumption of Lemma \ref{lemma:compound_lemma} is fulfilled. The second assumption 
(i.e. $D_{ij} \leq \max\{D_{ii}, D_{jj}\}$) holds by Lemma \ref{lemma:maximum_trace}. 
Using Lemma \ref{lemma:compound_lemma}, we obtain
\begin{align}
  \int_{\mathfrak{U}(\hr_A)} Q_v \ dv \leq 2 \left(\frac{L}{d_A} + 2 \sum_{i=1}^N 
  \sqrt{L\cdot \rank(\rho_{E,i})\|\rho_{B,i} \|_2^2} \right).\nonumber
\end{align}
Note that we replaced $\|\rho_{AE,i}\|_2$ by $\|\rho_{B,i}\|_2$ for every $i$, which is admissible, because they are 
complementary marginals of a pure state \cite{araki70}.
\end{proof}
\begin{corollary}\label{corollary1}
 Lemma \ref{lemma:expectation} provides the desired bound on the worst-case merging fidelity for finite sets. If we choose
$\mathcal{M}$ to be composed of the $L$-instrument $\mathcal{A}(\tilde{v})$ for some $\tilde{v}$ which fulfills the bound 
on the right hand side of (\ref{lemma:expectation_1}), and $\{\mathcal{U}_k\}_{k=1}^D$ which is a maximizer realizing 
$F_m(\overline{\rho}_{AB},\mathcal{A}(\tilde{v}))$ for $\tilde{v}$ (see eq. (\ref{max_fidelity_defined})), we have 
\begin{align*}
 F_m(\overline{\rho}_{AB}, \mathcal{M}) \geq 1 -  2 \left(\frac{L}{d_A} + 2 \sum_{i=1}^N 
  \sqrt{L\cdot \rank(\rho_{E,i})\|\rho_{B,i} \|_2^2} \right)
\end{align*}
which implies, together with the convexity property of $F_m$ (see Lemma \ref{merg_fid_rep}),
\begin{align*}
\underset{i \in [N]}{\min} F_m(\rho_{AB,i}, \mathcal{M}) \geq 1 -  2N \left(\frac{L}{d_A} + 2 \sum_{i=1}^N 
  \sqrt{L\cdot \rank(\rho_{E,i})\|\rho_{B,i} \|_2^2} \right).
\end{align*}
\end{corollary}

\section{\label{merging_proof}Proof of the merging theorem}
\begin{subsection}{Typical subspaces}
Here we state some properties of frequency typical projections which will be needed in the achievability 
proof. The concept of typicality is standard in classical and quantum information theory. Therefore we provide just the 
needed properties which can be found (along with basic definitions) in Ref. \cite{bjelakovic08c} (see Ref.\cite{csiszar11} for 
the properties of types and typical sequences). 
\begin{lemma}\label{types}
There exists a real number $c>0$ such that for every Hilbert space $\hr$ of dimension $d$ the following holds:
For each state $\rho$ on $\hr$, $\delta \in (0,\frac{1}{2})$ and $l \in \nn$ there is a projection $q_{\delta,l} 
\in \bo(\hr^{\otimes l})$ (its so-called \emph{frequency typical projection)} with
\begin{enumerate}
 \item $\tr(q_{\delta,l}\rho^{\otimes l}) \geq 1 - 2^{-l(c\delta^2-h(l))}$
 \item $q_{\delta,l}\rho^{\otimes l}q_{\delta,l} \leq 2^{-l(S(\rho)-\varphi(\delta))}q_{\delta,l}$
 \item $2^{l(S(\rho)-\varphi(\delta)-h(l))} \leq \rank(q_{\delta,l})\leq 2^{l(S(\rho)+\varphi(\delta))}$
\end{enumerate}
where the functions $\varphi(\delta)\rightarrow 0$ for $\delta \rightarrow 0$ and $h(l) \rightarrow 0$ for  
$l \rightarrow \infty$. Explicitly they are given by
\begin{align}
 h(l) = \frac{d}{l}\log(d+1)\hspace{0.2cm} \text{and} \hspace{0.4cm}  \varphi(\delta) = - \delta &\log\frac{\delta}{d} \nonumber 
\end{align}
for all $l \in \nn$ and $\delta \in (0,\frac{1}{2})$.
\end{lemma}
\end{subsection}

\subsection{\label{subsect:finite}Proof of the direct part in case of finite sets of states}
In this section we prove the optimal merging rate theorem using our one-shot result from Lemma \ref{lemma:expectation}. 
We first consider a finite set 
$\mathcal{X}:= \{\rho_{AB,i}\}_{i=1}^N \subset \st(\hr_{AB})$ with purifications $\psi_{ABE,1},...,\psi_{ABE,N} \in \hr_{ABE}$. For these states we
introduce some sort of ``typical reductions''. We define 
\begin{align*}
 \tilde{\psi}_{ABE,i,\delta}^l := \frac{1}{\sqrt{w_{i,\delta,l}}} \tilde{q}_{i,\delta}^l \psi_{ABE,i}^{\otimes l},
\end{align*}
where $w_{i,\delta,l} := \tr(\tilde{q}_{i,\delta}^l\psi_{ABE,i}^{\otimes l})$,
\begin{align}
 \tilde{\rho}_{B,i,\delta}^l := \tr_{\hr_{AE}^{\otimes l}}(\tilde{\psi}_{ABE,i,\delta}^l),\; \text{and} \hspace{0.3cm}
 \tilde{\rho}_{E,i,\delta}^l := \tr_{\hr_{AB}^{\otimes l}}(\tilde{\psi}_{ABE,i,\delta}^l). \nonumber
\end{align}
for all $i \in \{1,...,N\}, l\in \nn$ and $\delta \in
(0,\frac12)$. Here $\tilde{q}_{i,\delta}^l$ is given by the typical
projectors $q_{A,i}$, $q_{B,i}$ and $q_{E,i}$ of the corresponding marginals of $\psi_{ABE,i}$ 
\begin{align*}
 \tilde{q}_i := q_{A,i} \otimes q_{B,i} \otimes q_{E,i}  
\end{align*}
(here and in the following, the indices $\delta, l, i$ are sometimes omitted for the sake of brevity).
The following lemma provides some bounds needed later
\begin{lemma}\label{error_type}
 With the definitions given above, we have
\begin{enumerate}
 \item $w_{i,\delta,l}\geq 1 - 4\cdot 2^{-l(c\delta^2 - h(l))}$
 \item $\|\tilde{\rho}_{B,i,\delta}^{l}\|_2 \leq w_{i,\delta,l}^{-1} 2^{-\frac{l}{2}(S(\rho_{B,i}) - 
	3 \varphi(\delta)-h(l))}$
 \item $\rank(\tilde{\rho}_{E,i,\delta}^{l}) \leq 2^{l(S(\rho_{AB,i})+ \varphi(\delta))}$ 
\end{enumerate}
for all $i \in \{1,...,N\}$, $\delta \in (0,\frac{1}{2})$ and $l \in \nn$.
\end{lemma}
Note, that the functions $\varphi$ and $h$ in Lemma \ref{types} depend on the dimensions of the individual Hilbert space, 
however the above lemma clearly holds if we take the functions $\varphi$ and $h$ in Lemma \ref{types} with 
$d=\dim(\hr_{ABE})$.
\begin{proof}
\emph{1.)} 
Some simple algebra shows that 
\begin{align*}
 \tilde{q}		&= \eins_{ABE} - q_A \otimes q_B^\perp \otimes \eins_E - q_A^\perp \otimes \eins_B \otimes q_E \\
			&\hphantom{\mathrel{=}}- \eins_A \otimes q_B \otimes q_E^\perp - q_A^\perp \otimes q_B^\perp 
			   \otimes q_E^\perp\\
			&\geq \eins_{ABE} - q_A^\perp \otimes \eins_{BE} - \eins_A \otimes  q_B^\perp \otimes \eins_E \\
			&\hphantom{\mathrel{\geq}}- 2(\eins_{AB} \otimes q_E^\perp) \nonumber
\end{align*}
holds. Therefore
\begin{align}
 w_{i,\delta,l} &= \tr(\tilde{q}_i\psi_{ABE,i}^{\otimes l}) \\
		&\geq 1 - \tr(q_{A,i}^\perp\rho_{A,i}^{\otimes l}) - 
		  \tr(q_{B,i}^{\perp}\rho_{A,i}^{\otimes l}) - 2\tr(q_{E,i}^{\perp}\rho_{A,i}^{\otimes l})\\
		&\geq 1 - 4 \cdot 2^{-(c\delta^2-h(l))}.
\end{align}
\emph{2.)}
We first show, that
\begin{align}
\tr \left(\tr_{\hr_{AE}^{\otimes l}}(\tilde{q}^l\psi_{ABE}^{\otimes l}\tilde{q}^l)^2\right) \leq \tr\left((q_{B}^l \rho_{B}^{\otimes
l}q_{B}^l)^2\right)  \label{types_inequality_1}
\end{align}
holds. Note, that 
\begin{align}
\tr_{\hr_{AE}^{\otimes l}}\left((q_A^l \otimes \eins_{\hr_B} \otimes q_E^l)  \psi_{ABE}^{\otimes l}(q_A^l \otimes 
      \eins_{\hr_B} \otimes q_E^l)\right) =
\tr_{\hr_{AE}^{\otimes l}}\left((q_A^l \otimes \eins_{\hr_B} \otimes q_E^l)\psi_{ABE}^{\otimes l}\right). 
    \label{types_inequality_2}
\end{align}
Additionally, we have
$\tr_{\hr_{AE}^{\otimes l}}\left((q_A^l \otimes \eins_{\hr_B} \otimes q_E^l)\psi_{ABE}^{\otimes l}\right) 
    \leq \rho_B^{\otimes l}$, because 
\begin{align*}
 \rho_{B}^{\otimes l} - \tr_{\hr_{AE}^{\otimes l}}\left((q_A^l \otimes \eins_{\hr_B} \otimes q_E^l)  
    \psi_{ABE}^{\otimes l} \right) 
   &=\tr_{\hr_{AE}^{\otimes l}}\left(q_A^{l\perp} \otimes \eins_{\hr_B} \otimes q_E^l  \psi_{ABE}^{\otimes l}\right) \\
   &+\tr_{\hr_{AE}^{\otimes l}}\left(q_A^l \otimes \eins_{\hr_B} \otimes q_E^{l\perp}  \psi_{ABE}^{\otimes l}\right) \\
   &+\tr_{\hr_{AE}^{\otimes l}}\left(q_A^{l\perp} \otimes \eins_{\hr_B} \otimes q_E^{l\perp}\psi_{ABE}^{\otimes l}\right),
\end{align*}
where all of the summands on the r.h.s. are nonnegative operators. Therefore
\begin{align*}
 \tr \left(\tr_{\hr_{AE}^{\otimes l}}(\tilde{q}_{i}^l\psi_{ABE,i}^{\otimes l}\tilde{q}_{i}^l)^2\right) 
 &= \tr \left(\left(q_B^l\tr_{\hr_{AE}^{\otimes l}}(q_A^l \otimes \eins_{\hr_B} \otimes q_E^l \psi_{ABE,i}^{\otimes l})q_B
      \right)^2\right)\\
 & \leq \tr \left((q_B^l \rho_B^{\otimes l} q_B^l) 
    (q_B^l\tr_{\hr_{AE}^{\otimes l}}(q_A^l \otimes \eins_{\hr_B} \otimes q_E^l \psi_{ABE,i}^{\otimes l})q_B^l)\right)\\
& \leq \tr \left((q_B^l \rho_B^{\otimes l} q_B^l )^2\right),
\end{align*}
which proves eq. (\ref{types_inequality_1}). The above inequalities rely on the fact, that $\tr(A(\cdot))$ and $q_B^l(\cdot)q_B^l$
are positive maps, if $A$ is a nonnegative operator.
Finally we arrive at
\begin{align}
 \|\tilde{\rho}_{B,i,\delta}^{l}\|_2^2 
  &= w_{i,\delta,l}^{-2} \tr \left(\tr_{\hr_{AE}^{\otimes l}}(\tilde{q}^l_{i,\delta}\psi_{ABE,i,\delta}^{\otimes l}
      \tilde{q}_{i,\delta}^l)^2\right) \nonumber \\
  &\leq w_{i,\delta,l}^{-2} \tr\left((q_{B,i,\delta}^l 
      \rho_{B,i,\delta}^{\otimes l}q_{B,i,\delta}^l)^2\right) \label{types_inequality_3}\\
  &\leq w_{i,\delta,l}^{-2}\tr(q_{B,i,\delta}^l)\cdot 2^{-2l(S(\rho_{B,i})-
    \varphi(\delta))}\label{types_inequality_4} \\
  &\leq w_{i,\delta,l}^{-2}2^{-l(S(\rho_{B,i})-3\varphi(\delta))} \nonumber
\end{align}
where the r.h.s. of eq. (\ref{types_inequality_3}) follows from (\ref{types_inequality_1}), and (\ref{types_inequality_4})  
results from Lemma \ref{types}.2 applied twice. The last of the above inequalities follows from 
Lemma \ref{types}.3 . \newline
\emph{3.)}
follows from the third claim in Lemma \ref{types} and the fact that $S(\rho_{AB,i}) = S(\rho_{E,i})$ holds.
\end{proof}

\begin{theorem}\label{theorem:finite_ach}
For a finite collection $\mathcal{X}:= \{\rho_{AB,i}\}_{i=1}^N$ of states on $\hr_{AB}$, it holds
\begin{align*}
C_m(\mathcal{X}) \leq \max_{1 \leq i \leq N} S(A|B;\rho_{AB,i}). 
\end{align*}
\end{theorem}

\begin{proof}
The proof is similar to the corresponding one given in Ref. \cite{horodecki07b}, but uses the one-shot bound given in Lemma 
\ref{lemma:expectation}. We show, that the for every $\epsilon > 0$, the number $\max_{i \in [N]} 
S(A|B;\rho_{AB,i}) + \epsilon$ is an achievable rate for a merging of $\mathcal{X}$. First assume, that 
$\max_{i \in [N]} S(A|B;\rho_{AB,i}) < 0$. Let $\delta \in (0,\frac{1}{2})$ such that $\frac{\epsilon}{5}
< \varphi(\delta)$. It suffices to consider $\epsilon$ with $0 < \epsilon < |\max_{1<i\leq N} S(A|B,\rho_{AB,i})|$.
Define
\begin{align*}
L_l := \left\lfloor \exp\left(-l\left(\max_{i \in [N]}S(A|B;\rho_{AB,i}) + \epsilon\right)\right)\right\rfloor. 
\end{align*}
According to Lemma \ref{lemma:expectation} along with Corollary \ref{corollary1}, there is an $L_l$-merging 
$\mathcal{M}_l$ which fulfills 
\begin{align*}
\min_{i \in [N]}F(\mathcal{M}_l\otimes id_{\hr_E^{\otimes l}}(\tilde{\psi}^{l}_{ABE,i,\delta}), 
\phi_{L_l} \otimes \tilde{\psi}^{l}_{B'BE,i,\delta}) \geq 1 - NQ
\end{align*}
with 
\begin{align}
Q:= 2 \left(\frac{L_l}{\dim(\hr_A^{\otimes l})} + 2 \sum_{i=1}^N \sqrt{L_l\cdot \rank(\tilde{\rho}^{l}_{E,i,\delta})
    \|\tilde{\rho}_{B,i,\delta}^{l} \|_2^2} \right). \label{error_term}
\end{align}
With help of Lemma $\ref{error_type}$ it is easy to bound the summands on the r.h.s. of eq. (\ref{error_term}). Explicitly 
it holds
\begin{align*}
\frac{L_l}{\dim(\hr_{A}^{\otimes l})}   & \leq \frac{L_l}{\tr(q_{A,i})}\leq 2^{-6l\varphi(\delta)},  \\
\sqrt{L_l\cdot \rank(\tilde{\rho}^{l}_{E,i,\delta})
    \|\tilde{\rho}_{B,i,\delta}^{l} \|_2^2} &\leq \frac{2^{-\frac{l}{2}\varphi(\delta)}}
      {1-4\cdot2^{-l(c\delta^2 - h(l))}}.
\end{align*}
Therefore 
\begin{align*}
 \min_{i \in [N]}F\left(\mathcal{M}_l\otimes id_{\hr_E^{\otimes l}}(\tilde{\psi}^{l}_{ABE,i,\delta}), 
\phi_{L_l} \otimes \tilde{\psi}^{l}_{B'BE,i,\delta}\right) \geq 1 - \tilde{f}(l, N, \delta)
\end{align*}
holds, where
\begin{align}
 \tilde{f}(l, N, \delta) :=  2N \left(2^{-6l\varphi(\delta)} + 2N\frac{2^{-\frac{l}{2}\varphi(\delta)}}
  {1-4\cdot2^{-l(c\delta^2 - h(l))}}\right) 
\end{align}
for $l, N \in \nn$ and $\delta \in (0, \frac{1}{2})$. 
The desired bound for the merging fidelity of the original set $\mathcal{X}$ of states follows from Winter's 
gentle measurement Lemma (cf. Ref. \cite{winter99a}, Lemma 9). Explicitly, it holds
\begin{align}
  \min_{i \in [N]}F(\rho_{AB,i}^{\otimes l},\mathcal{M}_l) \geq 1 - f(l, N, \delta). 
  \label{error_end}
\end{align}
where $f(l, N, \delta) := 2\sqrt{\tilde{f}(l,N,\delta)} - 2\sqrt{32\cdot 2^{-l(c \delta^2 -h(l))}}$.
It remains to consider the case $\max_{i \in \{1,...,N\}} S(A|B;\rho_{AB,i}) \geq 0$. The above argument can be 
used with additional assistance of a sufficient amount of entanglement shared by the merging partners. Let $\phi_K$ be 
a maximally entangled state shared by $A$ and $B$ of Schmidt rank 
$K := 2^{\lceil \max_{i \in [N]} S(A|B,\rho_{AB,i})\rceil + 1}$ then for every $i$ the state 
\begin{align}
 \phi_K \otimes \rho_{AB,i}  \nonumber 
\end{align}
has negative conditional von Neumann entropy. Therefore the above argument holds for these states giving an 
$L_l$-merging $\widetilde{\mathcal{M}}_l$ with  
\begin{align}
 L_l = \exp\left(-l\left(\max_{1 \leq i \leq N} S(A|B,\rho_{AB,i}) - \left\lceil \max_{1 \leq i \leq N}S(A|B,\rho_{AB,i})\right\rceil - 1 + \epsilon\right)\right)
\end{align}
and $\min_{i \in [N]} F_m((\phi_K \otimes \rho_{AB,i})^{\otimes l}, \mathcal{M}_l)$ is lower bounded by a function
as on the r.h.s. of eq. (\ref{error_end}). Some unitaries which rearrange the tensor factors do the rest. Because 
\begin{align}
 \frac{1}{l}\log\left(\frac{K^l}{L_l}\right) = \underset{i \in [N]}{\max} S(A|B,\rho_{AB,i}) + \epsilon + o(l^0)
\end{align}
we are done
\end{proof}

\subsection{\label{subsect:general}Proof of the direct part for arbitrary sets of states}
In this section we aim to show that the achievability part of Theorem \ref{merging_theorem} does hold for any 
arbitrary set $\mathcal{X}$ of states as well. This can be achieved by approximating $\mathcal{X}$ by a sequence of 
(finite) nets and using the result obtained in the previous sections. The argument parallels the one given in case of compound quantum channels in Ref. \cite{bjelakovic09d}. \newline
A $\tau$-net in $\st(\hr)$ is a finite set $\{\rho_i\}_{i=1}^N$ such that for each state $\rho$ on $\hr$ there is at 
least one $i \in \{1,...,N\}$ with $\|\rho - \rho_i\|_1 < \tau$. We find such a finite set for every $\tau>0$ due to 
compactness of $\st(\hr)$. For our proof we have to ensure, that we find $\tau$-nets with cardinality upper bounded 
in an appropriate sense. This is the claim of the next lemma, which is a special case of Lemma 2.6 in Ref.  \cite{mil80}.
\begin{lemma}\label{card_bound}
 For any $\tau \in (0,1]$ there is a $\tau$-net $\{\rho_i\}_{i=1}^N$ in $\st(\hr)$ with cardinality 
 \begin{align*}
  N \leq  \left(\frac{3}{\tau}\right)^{2d^2} 
 \end{align*}
\end{lemma}
\begin{proof}
 The proof is exactly the same as the one given in Ref. \cite{bjelakovic08c} with the sets and norms replaced by the ones which are treated here.
\end{proof}
Let  $\mathcal{X} \subseteq \st(\hr_{AB})$ be an arbitrary set of states on $\hr_{AB}$. For a 
$\frac{\tau}{2}$-net $\tilde{\mathcal{X}}_{\tau}$, which fulfills the bound given in Lemma \ref{card_bound}, i.e.
\begin{align}
 |\tilde{\mathcal{X}}_\tau| \leq \left(\frac{6}{\tau}\right)^{2d_{AB}^2} \nonumber 
\end{align}
where $d_{AB} := \dim(\hr_{AB})$, we define the set
\begin{align}
 \mathcal{X}_{\tau} := \{\rho_i \in \tilde{\mathcal{X}}_\tau: \exists \rho \in \mathcal{X} \;\text{with} \; 
			 \|\rho_i - \rho\|_1 < \frac{\tau}{2} \}. \label{approx_net}
\end{align}
The following lemma provides some statements concerning $\tau$-nets needed later.
\begin{lemma}\label{net_lemma}
 Let $\mathcal{X} \subseteq \st(\hr_{AB})$ be a set of bipartite states on $\hr_{AB}$ and $\mathcal{X}_\tau$, for 
$\tau \in (0, \frac{1}{e}]$, the set
defined in (\ref{approx_net}). It holds
\begin{enumerate}
 \item $|\mathcal{X}_\tau| \leq \left(\frac{6}{\tau} \right)^{2d^2_{AB}}$,
 \item For every $\rho \in \st(\hr_{AB})$ there is a state $\rho_{i}$ in $\mathcal{X}_\tau$ satisfying
      \begin{align*}
       \|\rho^{\otimes l} - \rho_{i}^{\otimes l} \|_1 < l \cdot \tau,
      \end{align*}
 \item $|\underset{\rho \in \mathcal{X}}{\sup} S(A|B,\rho) - \underset{\rho_i \in \mathcal{X_\tau}}{\max} S(A|B,\rho_i)| 
	\leq \tau + 2\cdot \tau \log\left(\frac{d_{AB}}{\tau}\right)$, and
 \item Let $\mathcal{M}$ be any merging operation for states on $\hr_{AB}$. Then
 \begin{align}
  \underset{\rho_{i} \in \mathcal{X}_\tau}{\min} F_m(\rho_{i}^{\otimes l}, \mathcal{M}) 
  \geq 1 - \epsilon \Rightarrow \underset{\rho \in \mathcal{X}}{\inf} F_m(\rho^{\otimes l}, \mathcal{M}) \geq
   1 - 2 
  \sqrt{\epsilon} - 4 \sqrt{l\cdot \tau}
 \end{align}
\end{enumerate}
\end{lemma}

\begin{proof}
The first assertion is obvious from the definition of $\mathcal{X}_\tau$ together with Lemma \ref{card_bound}. The argument which proves the second one 
is exactly the same as done in Ref. \cite{bjelakovic08c} for channels. The third claim is a consequence of Fannes' inequality. 
Namely, to every positive real number $\tau$ we find states $\rho'$ in $\mathcal{X}$ and $\rho_i$ in $\mathcal{X}_\tau$
such that
\begin{align}
\|\rho' - \rho_i\|_1 < \tau \label{netzabstand}
\end{align}
and 
\begin{align}
 S(A|B,\rho') \geq \sup_{\rho \in \mathcal{X}} S(A|B,\rho) - \tau. \label{entropieabstand}
\end{align}
Eq. (\ref{netzabstand}) implies 
\begin{align}
 S(A|B,\rho') - S(A|B,\rho_i) \leq 2 \tau \log\left(\frac{d_{AB}}{\tau}\right) \nonumber
\end{align}
via twofold application of Fannes inequality \cite{fannes73}. Therefore
\begin{align}
 \sup_{\rho \in \mathcal{X}} S(A|B,\rho) - \tau &\leq S(A|B, \rho') \\
						&\leq S(A|B,\rho_i) + 2 \tau \log(\frac{d_{AB}}{\tau}). 
\end{align}
which proves the assertion. To verify the last claim of the lemma we first fix a purification corresponding to every 
member of 
$\mathcal{X}_\tau$ (remember that we are free in our choice of the purifications). Let $\psi_{ABE,i}$ be a purification 
of $\rho_{AB,i}$ on $\hr_{ABE}$ for $1 \leq i \leq N$. Let $\rho_{AB}$ an arbitrary element of $\mathcal{X}$, then we 
find at least one element of $\mathcal{X}_\tau$ satisfying
\begin{align}
 \|\rho_{AB,i}-\rho_{AB}\|_1 < \tau. \label{spurnetzbedingung}
\end{align}
As a consequence of Uhlmann's theorem, there exists a purification $\psi_{ABE}$ of $\rho_{AB}$ on $\hr_{ABE}$ such that
\begin{align}
 F(\rho_{AB}^{\otimes l},\rho_{AB,i}^{\otimes l}) = F(\psi_{ABE}^{\otimes l}, \psi_{ABE,i}^{\otimes l}). 
\label{uhlmann_net}
\end{align}
Now let $\phi_0$ and $\phi_1$ the maximally entangled input and output states associated with $\mathcal{M}$, then
\begin{align}
 &F_m(\rho_{AB}^{\otimes l},\mathcal{M}) \\
 =\;&F(\mathcal{M} \otimes id_{\hr_E^{\otimes l}}(\phi_0 \otimes \psi_{ABE}^{\otimes l}), \phi_1 \otimes \psi_{B'BE}^{\otimes l})\\
 \geq\;&1 - \|\mathcal{M} \otimes id_{\hr_E^{\otimes l}}(\phi_0 \otimes \psi_{ABE}^{\otimes l})-\phi_1 \otimes \psi_{B'BE}^{\otimes l}
  \|_1 \label{tdt}
\end{align}
where the last inequality follows from the bound given in eq. (\ref{fvg_1}). By an application of the triangle 
inequality, the trace distance on the r.h.s. of eq. (\ref{tdt}) is upper bounded by 
\begin{align}
  \|\mathcal{M} \otimes id_{\hr_E^{\otimes l}}(\phi_0 \otimes \psi_{ABE}^{\otimes l})-\phi_1 \otimes \psi_{B'BE}^{\otimes l}\|_1 
 \leq  & \|\mathcal{M} \otimes id_{\hr_E^{\otimes l}}(\phi_0 \otimes \psi_{ABE,i}^{\otimes l})-
      \phi_1 \otimes \psi_{B'BE,i}^{\otimes l}\|_1 \nonumber \\
 +\;& \|\mathcal{M} \otimes id_{\hr_E^{\otimes l}}(\phi_0 \otimes (\psi_{ABE}^{\otimes l} - 
      \psi_{ABE,i}^{\otimes l}))\|_1 \nonumber \\
 +\;& \|\phi_1 \otimes (\psi_{B'BE}^{\otimes l}-\psi_{B'BE,i}^{\otimes l})\|_1. \label{fidelity_net_1}
\end{align}
By monotonicity of the trace distance under the use of channels and eq. (\ref{spurnetzbedingung}), each of the two last 
summands can be upper bounded by $\|\psi_{ABE,i}^{\otimes l} - \psi_{ABE}^{\otimes l}\|_1$, and
\begin{align}
\|\psi_{ABE,i}^{\otimes l} - \psi_{ABE}^{\otimes l}\|_1 
	    &\leq 2 \sqrt{1 - F(\rho_{AB,i}^{\otimes l}, \rho_{AB}^{\otimes l})} \label{fidelity_net_2_1} \\
	    &\leq 2 \sqrt{\|\rho_{AB,i}^{\otimes l} - \rho_{AB}^{\otimes l}\|_1} \nonumber \\
	    &\leq 2 \sqrt{l\tau} \label{fidelity_net_2_2}
\end{align}
holds. Eq. (\ref{fidelity_net_2_1}) is justified by (\ref{uhlmann_net}) along with the relation given in eq. 
(\ref{fvg_2}), and  (\ref{fidelity_net_2_2}) is by the second claim of the present lemma.
The first summand is upper bounded by 
\begin{align}
 \|\mathcal{M} \otimes id_{\hr_E^{\otimes l}}(\phi_0 \otimes \psi_{ABE,i}^{\otimes l})-
  \phi_1 \otimes \psi_{B'BE,i}^{\otimes l}\|_1 \leq 2 \sqrt{\epsilon} \label{fidelity_net_3}
\end{align}
again with eq. (\ref{fvg_2}) and the assumptions. Eqns. (\ref{fidelity_net_1}), (\ref{fidelity_net_2_2}) and 
(\ref{fidelity_net_3}) justify
\begin{align}
 \inf_{\rho_{AB} \in \mathcal{X}} F_m(\rho_{AB}^{\otimes l}, \mathcal{M}) \geq 1 - 2 \sqrt{\epsilon} - 4 \sqrt{l \tau}
\end{align}
\end{proof}

\begin{theorem}\label{theorem:general_ach}
Let $\mathcal{X} \subset \st(\hr_{AB})$ be a set of states on $\hr_{AB}$. For the merging cost of $\mathcal{X}$ it holds
\begin{align}
 C_m(\mathcal{X}) \leq \sup_{\rho \in \mathcal{X}} S(A|B,\rho).
\end{align}
\end{theorem}

\begin{proof}
We show that 
\begin{align*}
 \sup_{\rho \in \mathcal{X}}S(A|B,\rho) + \epsilon 
\end{align*}
is an achievable rate for every $\epsilon$ satisfying $0 < \epsilon < |\sup_{\rho \in \mathcal{X}}S(A|B,\rho)|$.
Fix $\tau \in (0, \frac{1}{e})$ for the moment and consider the corresponding set $\mathcal{X}_\tau$ given in 
(\ref{approx_net}) which approximates $\mathcal{X}$. According to the proof of Theorem \ref{theorem:finite_ach} we find, 
for $l$ large enough,  an $(l,k_l)$-merging with
\begin{align}
 k_l &\leq \exp\left(l\left(\max_{1\leq i \leq N_\tau}S(A|B,\rho_i) + \frac{\epsilon}{2}\right)\right) \nonumber \\
    &\leq \exp\left(l\left(\sup_{\rho \in \mathcal{X}} S(A|B,\rho) + \frac{\epsilon}{2} + \tau + 2 \tau \log \frac{\dim(\hr_{AB})}{\tau}\right)\right) , 
    \label{net_ende_1}
\end{align}
where the second inequality is from Lemma \ref{net_lemma}. Another consequence of Lemma \ref{net_lemma} 
is the inequality
\begin{align}
 \inf_{\rho \in \mathcal{X}} F_m(\rho^{\otimes l}, \mathcal{M}_l) \geq 1 - 2 \sqrt{f(l, N_\tau,\delta)} - 
      4 \sqrt{l \cdot \tau}. \label{net_ende_2}
\end{align}
If we now choose a sequence $\{ \tau_l\}_{l\in \nn}$ such that $\lim_{l \rightarrow \infty} \tau_l = 0$ and 
$\lim_{\rightarrow \infty} \sqrt{l \cdot \tau_l} = 0$ hold, and additionally $N_{\tau_l}$ is growing polynomially (which is possible because 
Lemma \ref{card_bound} holds), then 
(\ref{net_ende_1}) and (\ref{net_ende_2}) show that $\sup_{\rho \in \mathcal{X}}S(A|B,\rho) + \epsilon$ is achievable.
\end{proof}

\subsection{\label{subsect:converse}Proof of the converse part}
Because we have shown that any rate above the least upper bound of the entanglement costs of the members of $\mathcal{X}$ 
achievable, our converse follows immediately from the original converse for single states from Ref. \cite{horodecki07b}. 
The argument given there is based on the fact that entanglement measures must be monotone under LOCC operations along with
an application of Fannes' inequality. As the proof is carried out in detail there, we just 
extend the argument to our present case.\\
Let $\delta > 0$ and $\chi_{AB}$ a member of $\mathcal{X}$ which satisfies
\begin{align}
 S(A|B,\chi_{AB}) \geq \sup_{\rho \in \mathcal{X}} S(A|B,\rho) - \delta .
\end{align}
Following the argument of the single state converse, we arrive at  
\begin{align}
 \frac{1}{l}\log(k_l) &\geq S(A|B,\chi_{AB}) - g(l)2\sqrt{\epsilon}(1 - \log(2\sqrt{\epsilon})) \nonumber \\
		      &=\sup_{\rho \in \mathcal{X}} S(A|B,\rho) - \delta - g(l)2\sqrt{\epsilon}
			  (1 - \log(2\sqrt{\epsilon}))
\end{align}
with a function $g$ which is O(1) for $l \rightarrow \infty$. Therefore the entanglement cost of $\mathcal{X}$ is 
least $\sup_{\rho \in \mathcal{X}} S(A|B,\rho) - \delta$ for every $\delta > 0$. 

\section{\label{sect:classical_cost}Classical communication cost of state merging}
Having determined the optimal entanglement cost of a state merging process, we consider the classical cost of state 
merging in this section. By classical cost, we mean the rate of classical communication from $A$ to $B$, which is at 
least required for an asymptotically perfect merging process. More precisely, if $\{\mathcal{M}_l\}_{l=1}^\infty$ is a 
sequence of $A \rightarrow B$ one-way LOCCs for a set $\mathcal{X}$, where $A$ distinguishes a number of $D_l$ measurement outcomes (see Section 
\ref{sect:definitions}, eq. (\ref{locc_allg})) within the application of $\mathcal{M}_l$, the classical cost is given by
\begin{align*}
 R_c = \limsup_{l \rightarrow \infty} \frac{1}{l}  \log D_l .
\end{align*}
In case of a single state $\rho_{AB}$, the minimum rate of classical communication for merging protocols achieving entanglement 
rate $R_q = S(A|B,\rho_{AB})$ was determined in Ref. \cite{horodecki07b} as 
$R_c = I(A;E, \rho_{AE})$, where $\rho_{AE}$ is the marginal on the subsystems belonging to $A$ and $E$ of an arbitrary 
purification $\psi_{ABE}$ of $\rho_{AB}$. In this section we deal with the case of a set of states to be merged and for
the sake of simplicity, we restrict ourselves to finite sets of states. 
Clearly, the classical communication cost of a merging procedure for a set $\mathcal{X}$ of states is lower bounded by 
the maximum of the communication costs for the individual states in $\mathcal{X}$. This is a direct consequence
of the known result for single states, which was given in Ref. \cite{horodecki07b}. The original proof given there is 
based on properties of the closely related ``mother protocol'' \cite{abeyesinghe09} and general assertions within the 
resource framework from Ref. \cite{devetak08}. Here, we give a more elementary proof for the reader not familiar with the 
results of Refs. \cite{devetak08} and \cite{abeyesinghe09}. Moreover, this result and a converse statement for the case that $A$ 
and $B$ are restricted to $L$-mergings show, that the protocol class we considered to show achievability of the merging 
cost, is suboptimal regarding the classical cost. 

\begin{proposition}[cf. Ref. \cite{horodecki07b}, Theorem 8]\label{single_class_converse}
Let $\rho_{AB} \in \st(\hr_{AB})$ be a bipartite state with purification $\psi_{ABE}$ on a space $\hr_{ABE}$ and 
$\epsilon \in (0,1)$. If $\mathcal{M}(\cdot) := \sum_{k=1}^D \mathcal{A}_k \otimes \mathcal{B}_k(\cdot)$ 
is an $A  \rightarrow B$ one-way LOCC such that
\begin{align}
 F(\mathcal{M}\otimes id_{\hr_E^{\otimes l}}(\phi_K \otimes \psi_{ABE}^{\otimes l}), \phi_L \otimes 
    \psi_{B'BE}^{\otimes l}) \geq 1 - \epsilon
\label{single_class_converse_error}
\end{align}
holds with maximally entangled states $\phi_K , \phi_L$ of Schmidt rank $K$ resp. $L$, then
\begin{align}
\frac{1}{l}\log(D) \geq I(A;E, \rho_{AE}) - 6 \sqrt{\epsilon}\left(\frac{1}{l}\log(KL) + \log \dim \hr_{AB}\right)
      - 3\eta(2\sqrt{\epsilon}) \label{single_class_converse_bound}
\end{align}
holds, where the function $\eta$ is defined on $[0,1]$ by
\begin{align}
 \eta(x) := \begin{cases} -x\log x & 0 < x \leq \frac{1}{e}  \\  
			  \frac{\log e}{e}& \frac{1}{e} < x \leq 1 \end{cases} \label{fannes_eta}
\end{align}
and $\eta(0) := 0$.
\end{proposition}

\begin{proof}
The proof is inspired by ideas from Ref. \cite{groisman05}. Fix $\epsilon \in (0,1)$ and $l \in \nn$. Let $\phi_K \in 
\kr_{AB}^0$ and $\phi_L \in \kr_{AB}^1$ maximally entangled input resp. output states of the protocol such that 
with notations
\begin{align*}
 \psi_0 := \phi_K \otimes \psi_{ABE}^{\otimes l},\  \text{and} 
  \hspace{0.4cm} \psi_1 := \phi_L \otimes \psi_{B'BE}^{\otimes l}
\end{align*}
eq. (\ref{single_class_converse_error}) reads
\begin{align}
F(\mathcal{M} \otimes id_{\hr_E^{\otimes l}}(\psi_0), \psi_{1}) \geq 1 - \epsilon. \label{single_class_converse_error_2}
\end{align}
We use the abbreviations $\hr^0_{BE} := \kr_B^0 \otimes \hr_{BE}^{\otimes l}$, $p_k := \tr(\mathcal{A}_k\otimes
id_{\hr^0_{BE}}(\psi_0))$ for 
$k \in [D]$, and $T = \{k \in [D]: p_k \neq 0\}$. It is well known, that the von Neumann entropy is an almost convex function, i.e.
for a state $\overline{\rho}$ defined as a mixture $\overline{\rho} := \sum_{i=1}^N {p_i} \rho_i$ of quantum states,  
\begin{align}
 S(\overline{\rho}) \leq H(p_1,...,p_N) + \sum_{i=1}^N p_i S(\rho_i) \nonumber
\end{align}
holds, where $H(p_1,...,p_N)$ is the Shannon entropy of the probability distribution on $[N]$ given by $p_1,...,p_N$. 
Using this fact, we obtain the lower bound 
\begin{align}
 \log D &\geq H(p_1,...,p_D) \nonumber  \\
        &\geq S\left(\sum_{k\in T}\mathcal{A}_k \otimes id_{\hr^0_{BE}}(\psi_0)\right) - 
	  \sum_{k \in T} p_k S\left(\frac{1}{p_k} \mathcal{A}_k\otimes id_{\hr^0_{BE}}(\psi_0)\right) 
	  \label{single_class_converse_fast_konv}
\end{align}
on $\log D$.
We separately bound the terms on the r.h.s. of eq. (\ref{single_class_converse_fast_konv}). 
With definitions $\pi_{K,A} := \tr_{\kr_B^0}(\phi_K)$, $\pi_{K,B} := \tr_{\kr_A^0}(\phi_K)$ and $\pi_{L,A} 
:= \tr_{\kr_B^1}(\phi_L)$ (these are maximally mixed states of rank $K$ resp. $L$) and $\mathcal{A}(\cdot) 
:= \sum_{k\in T} \mathcal{A}_k(\cdot)$, we obtain
\begin{align}
 S\left(\sum_{k \in T}\mathcal{A}_k \otimes id_{\hr^0_{BE}}(\psi_0)\right) 
	 &\geq S(\pi_{K,B} \otimes \rho_{BE}^{\otimes l})
      - S(\mathcal{A}(\pi_{K,A} \otimes \rho_A^{\otimes l})) \label{single_class_converse_term_1_1}\\
	& \geq \log K + lS(\rho_{BE}) - \log L - \Delta_1(\epsilon) \label{single_class_converse_term_1_2}\\
        & = \log\frac{K}{L} - l S(\rho_A) - \Delta_1(\epsilon) \label{single_class_converse_term_1_3}
\end{align}
where $\Delta_1(\cdot):= 2 \sqrt{\cdot} \log(L) + \eta(2 \sqrt{\cdot})$. Here eq. (\ref{single_class_converse_term_1_1}) 
is by the Araki-Lieb inequality \cite{araki70}, and  eq. (\ref{single_class_converse_term_1_3}) is due to the fact that 
$S(\rho_A) = S(\rho_{BE})$ holds. Eq. (\ref{single_class_converse_term_1_2}) is justified as follows. Using 
the relation between fidelity and trace distance from (\ref{fvg_2}) along with the fact, that the latter is monotone under 
taking partial traces, 
(\ref{single_class_converse_error_2}) implies
\begin{align}
 \|\mathcal{A}(\pi_{K,A} \otimes \rho_A^{\otimes l}) - \pi_{L,A} \|_1 \leq 2\sqrt{\epsilon}.
\end{align}
This, via application of Fannes' inequality leads to 
\begin{align}
 S(\mathcal{A}(\pi_{K,A} \otimes \rho_{A}^{\otimes l})) \leq S(\pi_{L,A}) - 2 \sqrt{\epsilon}\log L - 
    \eta(2 \sqrt{\epsilon}),
\end{align}
where $\eta$ is the function defined in (\ref{fannes_eta}). To bound the second term on the r.h.s. of 
(\ref{single_class_converse_fast_konv}), we use Stinespring extensions of the individual trace decreasing channels 
which constitute $\mathcal{M}$. Let for each $k \in [D]$, 
\begin{align}
v_k: \kr_A^0 \otimes \hr_{A}^{\otimes l} \rightarrow \kr_A^1 \otimes \hr_{C'} \nonumber
\end{align}
be a Stinespring extension of $\mathcal{A}_k$ and 
\begin{align}
u_k : \kr_B^0 \otimes \hr_B^{\otimes l} \rightarrow \kr_B^{1} \otimes \hr_{B'B}^{\otimes l} 
\otimes \hr_{C''}
\end{align}
be a Stinespring extension of $\mathcal{B}_k$. Here $\hr_{C'}$ is a Hilbert space associated to $A$ 
and $\hr_{C''}$ belongs to $B$. We fix notations $\mathcal{V}_k(\cdot):=v_k(\cdot)v_k^\ast$ and $\mathcal{U}_k := 
u_k(\cdot)u_k^\ast$ and denote the normalized outputs of these extensions by
\begin{align}
 \gamma_k := \frac{1}{p_k} \mathcal{V}_k \otimes \mathcal{U}_k \otimes id_{\hr_E^{\otimes l}}(\psi_0) 
    \label{single_class_kosten_gamma_k}
\end{align}
for every $k \in T$. Note that $\mathcal{V}_1,...,\mathcal{V}_D$ are trace decreasing, while $\mathcal{U}_1,...,
\mathcal{U}_D$ are channels. For every $k \in T$, we have
\begin{align}
 S\left(\tfrac{1}{p_k}\mathcal{A}_k \otimes id_{\hr^0_{BE}}(\psi_0)\right) &= 
  S\left(\tfrac{1}{p_k}\tr_{\hr_{C'}}\mathcal{V}_k\otimes id_{\hr^0_{BE}}(\psi_0)\right) \nonumber \\
 &=S\left(\tfrac{1}{p_k}\tr_{\hr_{C'}}\mathcal{V}_k\otimes \mathcal{U}_k \otimes id_{\hr_{E}^{\otimes l}}
    (\psi_0)\right)\nonumber \\
 &= S(\tr_{\hr_{C'}}\gamma_k), \label{single_klass_kosten_erstens}
\end{align}
where the second equality is by the fact that $u_k$ is an isometry and consequently the action of $\mathcal{U}_k$ does 
does not change the entropy.
Note, that (\ref{single_class_converse_error_2}) implies, 
because fidelity is linear in the first input here, existence of a positive number $c_k$ for every 
$k \in T$, such that 
\begin{align}
F\left(\frac{1}{p_k}\mathcal{A}_k \otimes \mathcal{B}_k \otimes id_{\hr_E^{\otimes l}}(\psi_0), \psi_1\right) = 1 - c_k 
      \label{single_class_converse_stine_fid}
\end{align}
and $\sum_{k\in T} p_k c_k \leq \epsilon$ hold. Because $\gamma_k$ is a purification of 
$\frac{1}{p_k}\mathcal{A}_k \otimes \mathcal{B}_k \otimes id_{\hr_E^{\otimes l}}(\psi_0)$ and $\psi_1$ is already pure,
Uhlmann's Theorem ensures existence of a pure state $\varphi_k$ on $\hr_{C'} \otimes \hr_{C''}$ with
\begin{align}
 F(\gamma_k, \psi_1 \otimes \varphi_k) &= \max\{|\braket{\gamma_k, \sigma}|^2 : \sigma \ \text{purification of}\ \psi_0 \ 
		    \text{on}\ 
			  \kr_{AB}^1 \otimes \hr_{B'BE}^{\otimes l} \otimes \hr_{C'}\otimes \hr_{C''} \} \nonumber \\
		     &= F\left(\frac{1}{p_k}\mathcal{A}_k \otimes \mathcal{B}_k \otimes id_{\hr_E^{\otimes l}}(\psi_0), 
			  \psi_1\right)
			\label{single_class_converse_uhl}
\end{align}    
for every $k \in T$. From eqns. (\ref{single_class_converse_stine_fid}) and (\ref{single_class_converse_uhl}) we conclude,
again via the well known relation between fidelity and trace distance from (\ref{fvg_2}),
\begin{align}
 \|\gamma_k - \psi_1 \otimes \varphi_k\|_1 \leq 2 \sqrt{c_k}, \label{single_class_converse_kleinklein}
\end{align}
which implies, again via Fannes' inequality and  monotonicity of the trace distance under partial tracing
\begin{align}
 S(\tr_{\hr_{C'}}\gamma_k) & \leq S(\psi_1 \otimes \tr_{\hr_{C'}}\varphi_k) +  \Delta_2(c_k) \nonumber \\
		           & \leq S(\tr_{\hr_{C'}}\varphi_k) + \Delta_2(c_k). 
				    \label{single_klass_kosten_fannes_anwendung_2}
\end{align}
where $\Delta_2(\cdot) = 2 \sqrt{\cdot} \log(\dim \hr_{AB}^2\dim\hr_{C''}) + \eta(2 \sqrt{\cdot})$. Consequently, we have
\begin{align}
  \sum_{k \in T} p_k S\left(\frac{1}{p_k}\mathcal{A}_k \otimes id_{\hr^0_{BE}}(\psi_0)\right)
  &=\sum_{k\in T} p_k S(\tr_{\hr_{C'}}\gamma_k) \nonumber \\
 &\leq \sum_{k \in T} p_kS(\tr_{\hr_{C'}}\varphi_k) + \Delta_2(\epsilon). 
     \label{single_class_converse_bound_delta_2}
\end{align}
The above equality is by (\ref{single_klass_kosten_erstens}), the inequality follows by 
(\ref{single_klass_kosten_fannes_anwendung_2}) and the fact, that that $\Delta_2$ is monotone and concave (see the 
definition of $\eta$ in \ref{fannes_eta})). It remains
to bound $\sum_{k \in T} p_kS(\tr_{\hr_{C'}}\varphi_k)$. Abbreviating 
$\hr_{AE}^1 := \kr_A^1 \otimes \hr_E^{\otimes l} \otimes \hr_{C'}$, an argument very similar to the one above gives 
(again via (\ref{single_class_converse_kleinklein}) and an application of Fannes' inequality) the bound
\begin{align}
 S(\tr_{\hr_{AE}^1}(\gamma_k)) & 
\geq S(\tr_{\hr_{AE}^1}(\psi_1 \otimes \varphi_k)) - \Delta_3(c_k) \nonumber\\
&= S(\pi_{L,B} \otimes \rho_{B'B}^{\otimes l} \otimes \tr_{\hr_{C'}}\varphi_k) - \Delta_3(c_k) 
    \label{single_class_converse_bound_delta_3}
\end{align}
with the function $\Delta_3(\cdot) := 2\sqrt{\cdot}(\log(K) + l\log(\dim\hr_{AB}\cdot \dim\hr_{C''})) 
+ 2 \eta(\sqrt{\cdot})$. And, using monotonicity and concavity of $\Delta_3$ together with 
(\ref{single_class_converse_bound_delta_3}), we obtain
\begin{align}
 \sum_{k \in T} p_k S(\tr_{\hr_{AE}^1}(\gamma_k)) &
 \geq \log(L) + lS(\rho_{AB}) + \sum_{k\in T}p_k S(\tr_{\hr_{C'}}\varphi_k) - \Delta_3(\epsilon) 
 \label{single_class_converse_deltarechnerei_1}
\end{align}
where we used, that $S(\rho_{B'B}) = S(\rho_{AB})$ holds.
If we now look at $\sum_{k=1}^D \mathcal{V}_k \otimes \mathcal{U}_k \otimes id_{\hr_E^{\otimes l}}(\cdot)$ as an 
one-way LOCC-channel with local operations on systems belonging to $A$ and $E$ on one side and $B$ on the 
other side which 3 the pure input state $\psi_0$ to the state described by the pure state mixture 
$\sum_{k\in T} p_k \gamma_k$, we have
\begin{align}
 S(\pi_K \otimes \rho_B^{\otimes l}) &= S(\tr_{\kr_{A}^0 \otimes \hr_{AE}^{\otimes l}}\psi_0)  \nonumber \\
				     &= S\left(\tr_{\hr_{AE}^1}\left(\sum_{k=1}^D \mathcal{V}_k 
					  \otimes id_{\hr_{BE}^0}(\psi_0)\right)\right) \nonumber \\
				     &\geq \sum_{k \in T} p_k \, S\left(\frac{1}{p_k}\tr_{\hr_{AE}^1}\mathcal{V}_k \otimes 
					    id_{\hr_{BE}^0}(\psi_0)\right) \nonumber \\
				     &= \sum_{k \in T} p_k\, S\left(\frac{1}{p_k}\tr_{\hr_{AE}^1}\mathcal{V}_k \otimes 
					    \mathcal{U}_k\otimes id_{\hr_E^{\otimes l}}(\psi_0)\right) 
				        \label{single_class_kosten_locc_mon_pre} \\
				      &= \sum_{k \in T} p_k \, S\left(\tr_{\hr_{AE}^1}\gamma_k\right).
					\label{single_class_kosten_locc_mon}
\end{align}
The second of the above equalities is due to the fact, that $\sum_{k=1}^D \mathcal{V}_k(\cdot)$ is trace preserving, the 
inequality is by concavity of the von Neumann entropy. Eq. (\ref{single_class_kosten_locc_mon_pre}) is because the von 
Neumann entropy is not changed by application of unitary channels in the input. The last equality is by the definitions 
introduced in (\ref{single_class_kosten_gamma_k}). With (\ref{single_class_converse_deltarechnerei_1}), 
(\ref{single_class_kosten_locc_mon}) and the equality $S(\rho_{AB}) = S(\rho_E)$, we obtain
\begin{align}
  S(\pi_K \otimes \rho_{B}^{\otimes l}) 
  & \geq \log(L) + lS(\rho_{E}) +  \sum_{k \in T} p_k S(\tr_{\hr_{C'}}\varphi_k) - \Delta_3(\epsilon). 
\label{single_class_converse_delta_4}
\end{align}
Rearranging the terms in inequality (\ref{single_class_converse_delta_4}) and using 
(\ref{single_class_converse_bound_delta_2}) leads to the bound
\begin{align}
 \sum_{k \in T} p_k S\left(\tfrac{1}{p_k}\mathcal{A}_k \otimes id_{\hr^0_{BE}}(\psi_0)\right) \leq \log \frac{K}{L} + 
      l(S(\rho_{AE}) - S(\rho_E))+ \Delta_2(\epsilon)+\Delta_3(\epsilon). \label{single_class_converse_term_2}
\end{align}
Here, we additionally used the fact, that $S(\rho_B) = S(\rho_{AE})$ holds. Combining the bounds from 
(\ref{single_class_converse_term_1_2}) and (\ref{single_class_converse_term_2}) with 
(\ref{single_class_converse_fast_konv}), we arrive at 
\begin{align}
 \frac{1}{l} \log D \geq I(A;E,\rho_{AE})-\frac{1}{l} 
   (\Delta_1(\epsilon) + \Delta_2(\epsilon) + \Delta_3(\epsilon)).
\end{align}
In fact, we find Stinespring extensions on spaces $\hr_{C'}$ and $\hr_{C''}$ with
\begin{align}
\dim \hr_{C'} &= K\cdot L \cdot \dim \hr_A^l\\
\dim \hr_{C''} &= K \cdot L \cdot \dim \hr_B^{2l} \dim \hr_A^l.
\end{align}
Using the definition of $\Delta_1, \Delta_2$ and $\Delta_3$ with the above dimensions, we conclude
\begin{align}
\frac{1}{l} \log D \geq I(A;E,\rho_{AE}) - 6 \sqrt{\epsilon}\left(\frac{\log KL}{l} + \log \dim \hr_{AB}\right) 
 - 3\eta(2\sqrt{\epsilon}),
\end{align}
which we aimed to prove.
\end{proof}
\begin{remark}

It is worth noting here, that the lower bound for the classical cost established in the proof of Proposition 
\ref{single_class_converse} does not explicitly rely on the entanglement rate of the protocol. 
Consequently, there is no chance to significantly reduce the required classical communication by 
admitting a higher entanglement rate, as long as one demands the protocol to be asymptotically perfect.
\end{remark}

In contrast to the above result, the following lemma indicates the limitations of the class of protocols used for 
establishing the achievability of the merging cost.

\begin{lemma}\label{l_merging_class_converse}
Let $\{\rho_{AB,i}\}_{i=1}^N$ be a set of states on $\hr_{AB}$. For every $\epsilon \in (0,1)$ and $\delta > 0$, there exists a number $l_0(\epsilon, \delta)$, such that
if $l > l_0$ and $\mathcal{M}(\cdot) := \sum_{k=1}^D \mathcal{A}_k \otimes \mathcal{B}_k(\cdot)$ is an 
$L$-merging for states on $\hr_{AB}^{\otimes l}$ for some $L \in \{1,...,\dim(\hr_{A}^{\otimes l})\}$ with 
\begin{align}
 \min_{1 \leq i \leq N} F_m(\rho_{AB,i}^{\otimes l}, \mathcal{M}) \geq 1 - \epsilon,
\end{align}
then
\begin{align}
 \frac{1}{l}\log(D) \geq \max_{1 \leq i \leq N} S(\rho_{A,i}) +  \frac{1}{l}\log\frac{K}{L} - \delta 
 \label{l_merging_class_unterbound}
\end{align}
holds.
\end{lemma}

\begin{proof}
First we consider for an arbitrary but fixed number $l \in \nn$ and an arbitrary single state $\rho_{AB}$. 
Let $M \subset [D]$ be a set of indices which fulfills
\begin{align*}
 F\left(\sum_{k \in M} \mathcal{A}_k \otimes \mathcal{B}_k \otimes id_{\hr_E^{\otimes l}}(\phi_K \otimes 
	\psi_{ABE}^{\otimes l}), \phi_L \otimes \psi_{B'BE}^{\otimes l} \right) \geq 1 - \epsilon
\end{align*}
We use abbreviations 
\begin{align*}
 \psi_0 := \phi_K \otimes \psi_{ABE}^{\otimes l} \hspace{0.3cm} \text{and} \hspace{0.3cm} 
    \rho_0 := \tr_{\kr_B^0 \otimes \hr_{BE}^{\otimes l}}(\psi_0) = \pi_K \otimes \rho_{A}^{\otimes l}
\end{align*}
Without any loss we assume that $M$ contains no index $k$ with $\tr(\mathcal{A}_k(\rho_0)) = 0$. Because we are concerned 
with an $L$-merging for $\psi_0$ here, we have
\begin{align*} 
 \mathcal{A}_k(\cdot) = u_k p_k(\cdot)p_k u_k^\ast  \hspace{0.3cm}
\end{align*}
for every $k$ in $M$ where $\{p_k\}_{k \in M}$ is a set of mutually orthogonal projections of rank $L$. We have
\begin{align*}
 \tr(\mathcal{A}_k (\rho_0)) = \tr(p_k\rho_0),
\end{align*}
and
\begin{align}
\tr(q \rho_0) = \sum_{k \in M}\tr(\mathcal{A}_k(\rho_0)), \label{single_class_lmerging_last}
\end{align}
where we used the definition $q := \sum_{k\in M} p_k$. It holds
\begin{align}
 1 - \epsilon &\leq \sum_{k \in M} F(\mathcal{A}_k \otimes \mathcal{B}_k \otimes id_{\hr_E^{\otimes l}}
			(\psi_0), \phi_L\otimes \psi_{B'BE}^{\otimes l}) \nonumber \\
	      &\leq \sum_{k \in M} F(\mathcal{A}_k(\rho_0), \pi_L) \label{single_class_lmerging_mon}\\
	      &\leq \sum_{k \in M} \tr(\mathcal{A}_k(\rho_0)) \label{single_class_lmerging_hom} \\
	      &= \tr(q \rho_0). \label{single_class_lmerging_hom_2}
\end{align}
Here, (\ref{single_class_lmerging_mon}) follows from the monotonicity of the fidelity under partial traces, 
(\ref{single_class_lmerging_hom}) by the fact that it is homogeneous in its inputs. The last equality is 
by (\ref{single_class_lmerging_last}). We may w.l.o.g. assume, that $\rho_0$ is of the form $\phi_0^{\otimes l} \otimes \psi_{ABE}^{\otimes l}$ with some maximally entangled state $\phi_0$, otherwise one could add a maximally entangled system to achieve this. In this case, Eq. (\ref{single_class_lmerging_hom_2}) would hold with the projector $\eins \otimes q$ instead of $q$, and this can be done without changing in the asymptotic rates. The well known fact, that subspaces of large probability, asymptotically, cannot have dimension substantially smaller than the typical subspace (see Ref. \cite{csiszar11}, Lemma 2.14) guarantees
\begin{align}
 \frac{1}{l}\log\tr(q) \geq S(\pi_0) + S(\rho_A) - \delta
\end{align}
if $l$ is sufficiently large. If we take into account, that $q$ is a sum of $|M|$ mutually orthogonal projections of rank $L$ (i.e. $\tr(q) = L\cdot |M|$), we have
\begin{align}
\frac{1}{l} \log|M| \geq S(\rho_A) - \frac{1}{l}\log \frac{K}{L} - \delta.
\end{align}
If we now consider a set $\mathcal{X}:= \{\rho_{AB,i}\}_{i=1}^N$ and and repeat the above argument with sets $M_1,...,M_N$ 
for this case we arrive at
\begin{align}
 \frac{1}{l}\log D \geq \frac{1}{l} \log \max_{1 \leq i \leq N} |M_i| \geq \max_{1 \leq i \leq N} S(\rho_{A,i}) - 
      \frac{1}{l} \log\frac{K}{L} - \delta
\end{align}
which concludes our proof.
\end{proof}

\begin{theorem}[classical cost of L-merging]
Let $\mathcal{X} := \{\rho_{AB,i}\}_{i=1}^N$ be a set of bipartite states on $\hr_{AB}$ and $\delta > 0$. For a merging procedure, 
where $A$ and $B$ are restricted to $L$-mergings (together with adding some further input pure entanglement) and  entanglement rate
\begin{align}
R_q = \max_{1 \leq i \leq N} S(A|B,\rho_{AB,i}) + \delta
\end{align}
is achieved, the optimal rate of classical communication is
\begin{align*}
 R_c = \max_{1 \leq i \leq N} S(\rho_{A,i}) + \max_{1 \leq i \leq N} S(A|B,\rho_{AB,i}) + \delta.
\end{align*}
\end{theorem}

\begin{proof}
The converse statement follows directly from Lemma \ref{l_merging_class_converse}. If $\{\mathcal{M}_l\}_{l=1}^\infty$ is
a merging which fulfills the assumptions of the Theorem, then
\begin{align*}
F(\mathcal{M}_l \otimes id_{\hr_E}^{\otimes l}(\phi_{K_l} \otimes \psi_{ABE,i}^{\otimes l}), 
      \phi_{L_l} \otimes \psi_{B'BE}^{\otimes l}) \geq 1 - o(l^0)
\end{align*}
with maximally entangled states $\phi_{K_l}$ resp. $\phi_{L_l}$ of Schmidt ranks for $K_l$ an $L_l$ for every $i \in [N]$, $l \in \nn$, and 
\begin{align}
 \limsup_{l \rightarrow \infty}\ \frac{1}{l}\log\left(\frac{K_l}{L_l}\right) = \max_{1 \leq i \leq N} S(A|B,\rho_{AB,i}) + \delta \label{l_merging_class_rate}
\end{align}
hold. With (\ref{l_merging_class_rate}) and Lemma \ref{l_merging_class_converse} it follows
\begin{align*}
\limsup_{l \rightarrow \infty} \frac{1}{l} \log(D_l) \geq \max_{1 \leq i \leq N} S(\rho_{A,i}) + 
    \max_{1 \leq i \leq N} S(A|B,\rho_{AB,i}) + \delta.
\end{align*}
To prove achievability, we step back to Section \ref{subsect:protocol}. Because $A$ and $B$ are using an $L_l$-merging for every $l$, the 
distinct number of measurement results $A$ has to communicate to $B$ is given by
\begin{align*}
 D_l = \frac{\dim\hr_{A}^{\otimes l}}{L_l}.
\end{align*}
The argument in Section \ref{subsect:protocol} shows, that the desired quantum rate can be 
achieved by choosing $L$-mergings for the mixtures 
\begin{align*}
\overline{\rho}_{AB}^l := \frac{1}{N}\sum_{i=1}^N\tilde{\rho}_{AB,i}^l,
\end{align*}
where $\tilde{\rho}_{AB,i}^l$ is the $\frac{\delta}{2}$-typically reduced state for $\phi_{K} \otimes \rho_{AB,i}$ for every $l \in \nn$ 
some $\delta \in (0,\frac{1}{2})$. 
We can therefore assume  $\hr_A^{\otimes l}$ to be restricted to the support of $\tilde{\rho}_A^l$. Clearly, it holds
\begin{align*}
 \rank \overline{\rho}_A^l &\leq \sum_{i = 1}^N \rank \tilde{\rho}_{A,i}^l  \\
			    &\leq N \cdot \max_{1 \leq i \leq N} \rank \tilde{\rho}_{A,i}^l\\
			    &\leq N \cdot \exp\left(l \left(\max_{1 \leq i \leq N}S(\pi_K \otimes  \rho_{A,i}) + \frac{\delta}{2}\right)\right).
\end{align*}
Therefore 
\begin{align}
 D_l \leq    \frac{N}{L_l} \cdot \exp\left(l\left(S(\pi_{K}) + \max_{1 \leq i \leq N} S(\rho_{A,i}) + \frac{\delta}{2}\right)\right)
\end{align}
and
\begin{align}
\frac{1}{l} \log(D_l) &\leq   \max_{1 \leq i \leq N} S(\rho_{A,i}) + \log \frac{1}{l}\left(\frac{K_l}{L_l}\right)+ 
			  \frac{N}{l} + \frac{\delta}{2} \nonumber \\
             	       &\leq \max_{1 \leq i \leq N} S(\rho_{A,i}) + \max_{1 \leq i \leq N} S(A|B,\rho_{AB,i}) + \delta 
\end{align}
if $l$ is large enough.
\end{proof}
The converse statement in the preceding Theorem is more strict than the one given in Prop. \ref{single_class_converse}.
The following example shows, that there are sets $\mathcal{X}$, where the optimal classical cost is surely not achieved
by using $L$-mergings. However, here we achieve the desired classical rate just by simple modifications of the protocol.

\begin{example}
Consider the set $\{\rho_{AB,1}, \rho_{AB,2} \} \subset \st(\hr_{AB})$ consisting of two members $\rho_{AB,1}= \phi_L$ and 
$\rho_{AB,2} = \pi_{M} \otimes \pi_{M}$, where $\phi_L$ is a maximally entangled state of Schmidt rank $L$ on a subspace
of $\hr_{AB}$ and $\pi_M$ is the maximally mixed state. We assume, that $L>M$ and
\begin{align}
 \supp(\rho_{A,1}) \perp \supp(\rho_{A,2}) \hspace{0.3cm} \label{example_othogonality}
\end{align} 
holds. In this case, we have 
\begin{align}
 \max_{i=1,2} I(A;E, \rho_{AE,i}) &= S(\rho_{A,2}) + S(A|B, \rho_{AB,2}) \\
				  &< S(\rho_{A,1}) + S(A|B, \rho_{AB,2}) \\
				  &= \max_{i=1,2} S(\rho_{A,i}) + \max_{i=1,2} S(A|B,\rho_{AB,i}). 
				      \label{klass_kosten_gap}
\end{align}

Since the supports of the $A$-marginals are orthogonal,
$A$ can perfectly distinguish his parts of the states (using one copy) and therefore get state knowledge. 
The rest is done by tracing out remaining entanglement to make both mergings have the same entanglement cost.
\end{example}

\begin{section}{Applications}\label{sect:applications}
In this section we give some indications how the result obtained so far has impact on other problems in quantum 
Shannon theory. As an example we provide another achievability proof for the entanglement generating capacity of a 
compound quantum channel with uninformed users. The original proof\cite{bjelakovic09d} was based on an one-shot 
result for entanglement transmission, a closely related concept (actually their capacities were shown to be equal). 
Here we follow another line of reasoning, namely we use the close correspondence between the task of distilling 
entanglement from bipartite sources and generating entanglement over quantum channels. To this end we prove a compound 
version of the so-called hashing bound which is known as a prominent lower bound on distillable entanglement for 
perfectly known states \cite{devetak05c}. For convenience we restrict ourselves to the case of finite sets of states 
and finite compound channels. The results are easily generalized to arbitrary sets using approximation techniques as it 
was done in Sect. \ref{subsect:general}.

\begin{subsection}{Entanglement distillation under state uncertainty}
Following Ref. \cite{devetak05c}, we define a $(l,k_l)$-\emph{protocol for one-way distillation} of states on $\hr_{AB}$ as a 
combination of an instrument $\{\mathcal{A}_k\}_{k=1}^D \subset \mathcal{C}^{\downarrow}(\hr_A^{\otimes l}, \kr^l)$ and 
a set of quantum channels $\{\mathcal{B}_k\}_{k=1}^D \subset \mathcal{C}(\hr_B^{\otimes l}, \kr^l)$ of the form
\begin{align*}
 \mathcal{T} = \sum_{k=1}^D \mathcal{A}_k \otimes \mathcal{B}_k, 
\end{align*}
such that 
$\dim(\kr^l) = k_l$. For a set $\mathcal{X} \subset \st(\hr_{AB})$ of states on $\hr_{AB}$ a nonnegative number $R$ is an 
achievable (one-way) entanglement distillation rate, if there is a sequence $\{\mathcal{T}_l\}_{l=1}^\infty$ of 
$(l, k_l)$-entanglement distillation protocols such that 
\begin{enumerate}
 \item $\underset{l \rightarrow \infty}{\liminf} \frac{1}{l}\log(k_l) \geq R $
 \item $\underset{l \rightarrow \infty}{\lim} \underset{\rho \in \mathcal{X}}{\inf} F(\mathcal{T}_l (\rho^{\otimes l}), 
	\phi_l) = 1$
\end{enumerate}
where $\phi_l$ is a maximally entangled state on $\kr^l \otimes \kr^l$. The number
\begin{align*}
 D_{\rightarrow}(\mathcal{X}) := \sup\{R: R\; \text{is an achievable rate for one-way entanglement distillation}\}. 
\end{align*}
is called the \emph{(one way) entanglement capacity} of $\mathcal{X}$. The following lemma is a compound analog to 
Theorem 3.1 in Ref \cite{devetak05c}.

\begin{lemma}\label{compound_hashing}
Let $\mathcal{X}:=\{\rho_i\}_{i=1}^N \subset \st(\hr_{AB})$ be a (finite) set of bipartite states on $\hr_{AB}$. Then
\begin{align}
 D_{\rightarrow}(\mathcal{X}) \geq - \underset{1\leq i \leq N}{\max}\; S(A|B,\rho_i) \label{hashing_ineq}
 \end{align}
\end{lemma}
\begin{proof}
It suffices to consider the case of a set with $\max_{1 \leq i \leq N} S(A|B,\rho_i) < 0$, since rate $0$ can always be
achieved by using a trivial protocol which distills no entanglement at all.
Let $\mathcal{M} := \sum_{k=1}^D \mathcal{A}_k \otimes \mathcal{U}_k$ be an $L$-merging for $\mathcal{X}$ satisfying
\begin{align}
 \underset{1 \leq i \leq N}{\min} \;  F(\mathcal{M} \otimes id_{\hr_E^{\otimes l}}(\psi_{ABE,i}^{\otimes l}), 
  \phi_l \otimes \psi_{B'BE,i}^{\otimes l}) \geq 1 - \epsilon.
\end{align}
Then $\mathcal{T}(\cdot):= \sum_{k=1}^D \mathcal{A}_k \otimes \mathcal{R}_k(\cdot)$ with
$\mathcal{R}_k := \tr_{\hr_{B'BE}^{\otimes l}} \circ (\mathcal{U}_k\otimes id_{\hr_E^{\otimes l}})$ for every $k$ is a 
one-way entanglement distillation protocol for $\mathcal{X}$ satisfying
\begin{align}
 &\hphantom{\mathrel{\geq}} \; F(\mathcal{T}(\rho_i^{\otimes l}), 
   \phi_l) \nonumber \\
 &\geq F(\mathcal{M}\otimes id_{\hr_E^{\otimes l}}(\psi_{ABE,i}^{\otimes l}), 
 \phi_l \otimes \psi_{B'BE,i}^{\otimes l}) \label{part_trace_fid} \\
 &\geq \; 1 - \epsilon. \nonumber
\end{align}
for every $1 \leq i \leq N$. Eq. (\ref{part_trace_fid}) is justified by the fact that taking partial traces
cannot decrease fidelity. Following the proof of Theorem \ref{theorem:general_ach}, we find for $\epsilon > 0$ and $l \in \nn$ large enough an $L_l$-merging $\mathcal{M}_l$ for $\mathcal{X}$ such that 
\begin{align}
 L_l \geq \left\lfloor \exp\left(-l(\max_{1\leq i \leq N} S(A|B,\rho_i) + \epsilon + o(l^0))\right)\right\rfloor \label{rate_dist}
\end{align}
and 
\begin{align}
\underset{1 \leq i \leq N}{\min} \;  F(\mathcal{M}_l \otimes id_{\hr_E^{\otimes l}}(\psi_{ABE,i}^{\otimes l}), 
  \phi_l \otimes \psi_{B'BE,i}^{\otimes l}) \geq 1 - o(l^0). 
\end{align}
holds. Eqns. (\ref{part_trace_fid}) and (\ref{fid_dist}) give
\begin{align}
 \min_{1 \leq i \leq N}F(\mathcal{T}_l(\rho_i^{\otimes l}), \phi_l) \geq 1 - o(l^0). \label{fid_dist}
\end{align}
The achievability of $-\max_{1 \leq i \leq N}S(A|B,\rho_i)$ follows from (\ref{rate_dist}) and (\ref{fid_dist}).
\end{proof}
The above lemma provides the main building block for determining the one-way entanglement capacity for sets of states,
which is done in the following theorem.
\begin{theorem}\label{ent_dist_theorem}
 Let $\mathcal{X} := \{\rho_i\}_{i=1}^N \subset \st(\hr_{AB})$. Then
\begin{align}
 D_{\rightarrow}(\mathcal{X}) =
  \underset{l \rightarrow \infty}{\lim}\frac{1}{l} D^{(1)}(\mathcal{X}^{\otimes l}) \label{dist_capacity}
\end{align}
with 
\begin{align}
 D^{(1)}(\mathcal{X}) := 
-\underset{\mathcal{T}}{\min}\; \underset{1 \leq i \leq N}{\max}\sum_{j : \lambda_j^{(i)}\neq 0} 
    \lambda_j^{(i)} S(A|B, \rho_j^{(i)})   \label{d_1}
\end{align}
where the minimization is over quantum instruments $\mathcal{T}$ of the form $\mathcal{T}:=\{ \mathcal{T}_j\}_{j=1}^J$ on
$\hr_A$ with definitions 
\begin{align}
 \lambda_j^{(i)} := \tr(\mathcal{T}_j(\tr_{\hr_B}\rho_i)) \hspace{0.2cm} \text{and} \hspace{0.2cm} 
  \rho_{j}^{(i)} := \frac{1}{\lambda_{j}^{(i)}}\mathcal{T}_j \otimes id_{\hr_B}(\rho_i) \label{distillation_capacity}
\end{align}
for $1 \leq j \leq J$ and $1 \leq i \leq N$ with $\lambda_{j} \neq 0$. In fact, we can restrict ourselves to 
$J \leq dim(\hr_A)^2$ (see \cite{devetak05c}).
\end{theorem}

\begin{remark}
One easily verifies, that the limit in (\ref{dist_capacity}) exists. Clearly, 
\begin{equation}
 D^{(1)}(\mathcal{X}^{\otimes k}) + D^{(1)}(\mathcal{X}^{\otimes l}) \leq D^{(1)}(\mathcal{X}^{\otimes (k+l)})
\end{equation}
holds for any $k,l \in \nn$, because if $\mathcal{T}^{(k)}$ and $\mathcal{T}^{(l)}$ are instruments on 
$\hr_A^{\otimes k}$ resp. $\hr_A^{\otimes l}$, then $\mathcal{T}^{(k)}\otimes \mathcal{T}^{(l)}$ is an instrument on 
$\hr_{A}^{\otimes (k+l)}$. The rest is by Fekete's Lemma \cite{fekete23}.
\end{remark}

\begin{proof}[Proof of Theorem \ref{ent_dist_theorem}]
We begin with the direct part of the Theorem. Our proof parallels the one given in Ref. \cite{devetak05c} for the single 
state case. However, for the direct part, we use Lemma 
\ref{compound_hashing} instead of the single state hashing bound. To prove achievability,
let $\mathcal{T} := \{\mathcal{T}_j\}_{j=1}^J$ be any instrument on $\hr_A$, $\mathcal{P} := \{\mathcal{P}_j\}_{j=1}^J$ a 
set of channels of the form
\begin{align}
\mathcal{P}_j(\chi) := \chi \otimes \ket{e_j}\bra{e_j} 
\end{align}
for every $\chi \in \st(\hr_{B})$ and $1 \leq j \leq J$, where $e_1,...,e_J$ are members of an orthonormal basis of a 
Hilbert space $\hr_{B'}$ located at $B$'s site. Define states 
\begin{align*}
 \tilde{\rho}_i := \sum_{j=1}^J \mathcal{T}_j \otimes \mathcal{P}_j(\rho_i) 
 =\sum_{j, \lambda_j^{(i)} \neq 0} \lambda_j^{(i)} \rho_j^{(i)} \otimes \ket{e_j}\bra{e_j}
\end{align*}
for $1 \leq i \leq N$. These preprocessed states have conditional von Neumann entropy
\begin{align*}
S(A|BB', \tilde{\rho}_i) = \sum_{j: \lambda_j^{(i)} \neq 0} \lambda_j^{(i)} S(A|B, \rho_j^{(i)}).
\end{align*}
Direct application of Lemma \ref{compound_hashing} gives achievability. The converse statement can be proven just by the
same arguments as given in Ref. \cite{devetak05c}, we give the proof for convenience. We consider an
arbitrary $(l,k_l)$ one-way distillation protocol with rate $R$, given by a LOCC channel with $A \rightarrow B$ 
classical communication 
\begin{align*} 
 \mathcal{T}(\cdot) := \sum_{j=1}^J \mathcal{T}_j \otimes \mathcal{R}_j(\cdot)
\end{align*}
with $\mathcal{T}_j \in \mathcal{C}^{\downarrow}(\hr_A^{\otimes l}, \kr)$ and 
$\mathcal{R}_j \in \mathcal{C}(\hr_{B}^{\otimes l}, \kr)$ , $1 \leq j \leq J$, such that for a given 
$\tau \in (0,\frac{1}{2})$
\begin{align}
 F(\mathcal{T}(\rho_i^{\otimes l}), \phi) \geq 1 - \tau && (i \in \{1,...,N\}) \label{ent_dist_con_fid}
\end{align}
holds, where $\phi$ is a maximally entangled state on $\kr \otimes \kr$ and $\dim \kr = \lfloor 2^{l R} \rfloor$. We fix notations
\begin{align*}
 \lambda_j^{(i)} &:= \tr(\mathcal{T}_j \otimes \mathcal{R}_j(\rho_i^{\otimes l})), \hspace{0.3cm} \text{and} \hspace{0.3cm}
  \omega_j^{(i)} := \frac{1}{\lambda_j^{(i)}} \mathcal{T}_j \otimes R_j(\rho_i^{\otimes l}), \\
  \rho_{j}^{(i)} &:= \frac{1}{\lambda_j^{(i)}} \mathcal{T}_j \otimes id_{\hr_B^{\otimes l}}(\rho_i^{\otimes l})
\end{align*}
for $i \in [N]$,\ $j \in [J]$ with $\lambda_j^{(i)} \neq 0$. Application of $\mathcal{T}$ on $\rho_i$ results in the state
\begin{align*}
 \Omega^{(i)} := \sum_{j : \lambda_j^{(i)} \neq 0} \lambda_j^{(i)} \omega_j^{(i)}.
\end{align*}
Using the relation from (\ref{fvg_2}), (\ref{ent_dist_con_fid}) implies, that
\begin{align*}
 \|\Omega^{(i)}- \phi \|_1 \leq 2\sqrt{\tau} 
\end{align*}
holds for all $i \in [N]$, which leads us to 
\begin{align}
 | S(A | B, \Omega^{(i)}) - S(A | B, \phi) | \leq \epsilon \label{ent_dist_con_fannes}
\end{align}
with $\epsilon := 2(2 \sqrt{\tau} \log(\dim \kr^2) + \eta(2 \sqrt{\tau}))$ via twofold application of Fannes' inequality.
Eq. (\ref{ent_dist_con_fannes}) along with $S(A | B, \phi) = - l \cdot R$ implies
\begin{align}
 lR & \leq - S(A|B, \Omega^{(i)}) + 4 \sqrt{\tau} \cdot lR + 2 \eta(2 \sqrt{\tau}). \label{ent_dist_con_rate_1}
\end{align}
Moreover, we have
\begin{align}
  S(A|B,\Omega^{(i)})  
  & \geq \sum_{j:\lambda_j^{(i)}\neq 0} \lambda_j^{(i)} S(A|B, \omega_j^{(i)}) \nonumber \\
  & \geq \sum_{j:\lambda_j^{(i)}\neq 0} \lambda_j^{(i)} S(A|B, \rho_j^{(i)}),  \label{ent_dist_con_entr}
\end{align}
where the first inequality is by concavity of the map $\rho \mapsto S(A|B,\rho)$ for quantum states, the second is
by application of the quantum data processing inequality. Combining (\ref{ent_dist_con_rate_1}) and 
(\ref{ent_dist_con_entr}), we obtain
\begin{align*}
 lR  &\leq -\max_{i \in [N]}\sum_{j:\lambda_j^{(i)}\neq 0} \lambda_j^{(i)} 
	  S(A|B, \rho_j^{(i)}) + 4 \sqrt{\tau} l \cdot R + 2 \eta(2\sqrt{\tau}) \\
      &\leq -\min_{\mathcal{T}} \max_{i \in [N]} \sum_{j: \lambda_{j}^{(i)}\neq 0} \lambda_j^{(i)} 
	  S(A|B, \rho_j^{(i)}) + 4 \sqrt{\tau} l \cdot R + 2 \eta(2\sqrt{\tau}) \\
      &\leq D^{(1)}(\mathcal{X}^{\otimes l}) + 4 \sqrt{\tau} l \cdot R + 2 \eta(2 \sqrt{\tau})
\end{align*}
\end{proof}
\begin{remark}
Theorem \ref{ent_dist_theorem} shows, that one may have to pay an additional price for imperfect knowledge of the state. 
Namely, the capacity for a set $\mathcal{X}$ is, in general, strictly smaller than the minimum over the single-state 
capacities of the individual states in $\mathcal{X}$, as can be seen from eq. (\ref{d_1}). 
\end{remark}
\end{subsection}

\begin{subsection}{Entanglement generation over compound quantum channels}
Finally, we give another proof for the direct part of the coding theorem for entanglement generation over compound 
channels, which was originally given in Ref. \cite{bjelakovic09d}, Theorem 13. We first recall some definitions from Ref. \cite{bjelakovic09d}.
Let $\mathfrak{I}$ be a compound quantum channel generated by a set $\mathfrak{I}\subseteq \mathcal{C}(\hr_A,\hr_B)$ of
channels. We consider the uninformed user scenario, where precise knowledge about the identity of the channel
is available neither to encoder nor decoder. An \emph{entanglement generating $(l,k_l)$-code} for $\mathfrak{I}$ is a pair 
$(\mathcal{R}^l,\varphi^l)$ where $\mathcal{R}^l \in \mathcal{C}(\hr_B^{\otimes l}, \mathcal{K}^l)$ is a channel with 
$k_l = \dim \kr^l$ and $\varphi_l$ a pure state on $\kr^l \otimes  \hr_A^{\otimes l}$. A positive number $R$ is 
an \emph{achievable rate for entanglement generation} over $\mathfrak{I}$ if there is a sequence of $(l,k_l)$-entanglement 
generating codes satisfying
\begin{enumerate}
 \item $\underset{l \rightarrow \infty}{\liminf} \ \frac{1}{l} \log k_l \geq R$, and
 \item $\underset{l \rightarrow \infty}{\lim} \underset{\mathcal{N} \in \mathfrak{I}}{\inf}\ F(\phi_l, (id_{\kr^l} \otimes 
	\mathcal{R}^l\circ \mathcal{N}^{\otimes l})(\varphi_l)) = 1$, where $\phi_l$ denotes a maximally entangled state
	on $\kr^l \otimes \kr^l$.
\end{enumerate}
The number
\begin{align*}
 E(\mathfrak{I}) := \sup \{R : R  \; \text{is an achievable rate for entanglement generation over}\ \mathfrak{I} \}.
\end{align*}
is called the \emph{entanglement generating capacity} of $\mathfrak{I}$. 
\begin{theorem}[cf. Ref. \cite{bjelakovic09d}, Th. 13]\label{ent_generation}
Let $\mathfrak{I} := \{\mathcal{N}_i\}_{i =1}^N$ be a finite compound quantum channel, $\mathfrak{I} \subset 
\mathcal{C}(\hr_A, \hr_B)$. Then
\begin{align}
E(\mathfrak{I}) \geq \underset{l \rightarrow \infty}{\lim} \frac{1}{l}\underset{\rho \in \st(\hr_{A}^{\otimes l})}{\max} 
	\underset{1 \leq i \leq N}{\min} I_c(\rho, \mathcal{N}_i^{\otimes l}) \label{ent_gen_capacity}
\end{align}
holds
\end{theorem}
\begin{proof}
First note that the limit on the r.h.s of (\ref{ent_gen_capacity}) exists by standard arguments (see Ref. \cite{bjelakovic09d}, Remark 2). 
We just have to prove that the number 
\begin{align*}
\underset{1 \leq i \leq N}{\min} I_c(\rho,\mathcal{N}_i)-\epsilon 
\end{align*}
is an achievable rate for every state $\rho$ on $\hr_A$ and every $\epsilon>0$, the rest is by standard blocking 
arguments. There is nothing to prove for sets with $\underset{1 \leq i \leq N}{\min} I_c(\rho,\mathcal{N}_i)\leq 0$. 
Therefore let $\rho$ be a state on $\hr_A$ with $\min_{1 \leq i \leq N} I_c(\rho,\mathcal{N}_i)> 0$. 
Consider the set $\mathcal{X} := \{\rho_i\}_{i = 1}^N$ of bipartite states in $\hr_{AB}$, where $\rho_i$ is defined
\begin{align}
\rho_i := (id_{\hr_A} \otimes \mathcal{N}_i)(\chi) \label{def_comp_set}
\end{align} 
for $1 \leq i \leq N$. Here $\chi$ is the pure state on $\hr_A \otimes \hr_A$ such that the partial trace over any of 
the two subsystems results in the state $\rho$. We show that a good entanglement distillation protocol for 
the set $\mathcal{X}$ of bipartite states generated by $\mathfrak{I}$ implies the existence of 
a good entanglement generation code for $\mathfrak{I}$. Following the proof of Lemma \ref{compound_hashing},
there exists an $(l,k_l)$-distillation protocol $\mathcal{T} = \sum_{k=0}^D \mathcal{A}_k\otimes \mathcal{R}_k$ 
for $\mathcal{X}$ with $\mathcal{A}_k \in \mathcal{C}^{\downarrow}(\hr_A^{\otimes l}, \kr^l)$ and 
$\mathcal{R}_k \in \mathcal{C}(\hr_B^{\otimes l}, \kr^l)$ for $k \in \{1,...,D\}$ with $D$ determined by 
$\dim\hr_A$ and $\dim \kr^l$ such that 
\begin{align}
 \dim\kr^l \geq \left \lfloor \exp\left(l\left(\min_{1 \leq i \leq N} I_c(\rho, \mathcal{N}_i)-\epsilon\right)\right)\right \rfloor \label{ent_gen_rate} 
\end{align}
and
\begin{align}
 \min_{1 \leq i \leq N} F(\mathcal{T}(\rho_i),\phi_l) \geq 1 - o(l^0) \label{ent_gen_fid}
\end{align}
with $\phi_l$ being the maximally entangled state on $\kr^l$. Notice, that in eq. (\ref{ent_gen_rate}), we used the 
identity 
\begin{align*}
 I_c(\rho, \mathcal{N}_i) = - S(A|B,\rho_i) 
\end{align*}
for every $i \in \{1,...,N\}$.
The definitions given in eq. (\ref{def_comp_set}) imply
\begin{align*}
\mathcal{A}_k \otimes \mathcal{R}_k(\rho) = (id_{\kr^l} \otimes \mathcal{R}_k \circ \mathcal{N}_i)(\mathcal{A}_k\otimes 
    id_{\hr_A^{\otimes l}}(\chi))
\end{align*}
for every $0\leq k \leq D$ and $1 \leq i \leq N$. Therefore, 
\begin{align}
 F(\mathcal{T}(\rho_i),\phi_l) 
  &= \sum_{k=0}^D F(id_{\kr^l} \otimes \mathcal{R}_k\circ \mathcal{N}_i^{\otimes l}(\mathcal{A}_k\otimes 
      id_{\hr_A^{\otimes l}}(\chi)), \phi_l) \nonumber \\
 &= \sum_{k: p_k\neq 0} p_k F(id_{\kr^l} \otimes \mathcal{R}_k\circ \mathcal{N}_i^{\otimes l}(\varphi_k), \phi_l) 
      \label{av_ent_gen_code}
\end{align}
holds for every $i$, where we used the definitions
\begin{align*}
 p_k := \tr(\mathcal{A}_k(\rho)), \; \text{and} \hspace{0.7cm} 
    \varphi_k := \frac{1}{p_k}(\mathcal{A}_k \otimes id_{\hr_A^{\otimes l}})(\chi)
\end{align*}
for $p_k \neq 0$, $0\leq k \leq D$. Notice, that $\varphi_0,...,\varphi_D$ are pure states, because the operations 
$\mathcal{A}_k$ are pure since they arise from an $L$-merging (see the proof of Lemma \ref{compound_hashing}). 
Again because the fidelities are affine functions of the first input, (\ref{ent_gen_fid}) and (\ref{av_ent_gen_code}) 
imply 
\begin{align}
    \sum_{k:p_k\neq 0} p_k F\left(id_{\kr^l} \otimes \mathcal{R}_k\circ \frac{1}{N}\sum_{i=1}^N\mathcal{N}_i^{\otimes l}
    (\varphi_k), \phi_l\right) \geq 1 - o(l^0). \label{pre_end}
\end{align}
The r.h.s. of equation (\ref{av_ent_gen_code}) is, in fact, an average of
fidelities of entanglement generating codes $(\mathcal{R}_1,\varphi_1),...,(\mathcal{R}_D,\varphi_D)$ 
with probabilities $p_1,...,p_D$. This implies the existence of a number $k' \in \{1,...,D\}$ such that with 
$\varphi := \varphi_{k'}$ and $\mathcal{R} := \mathcal{R}_{k'}$ 
\begin{align}
 \min_{1 \leq i \leq N} F\left(id_{\kr^l} \otimes \mathcal{R}\circ \mathcal{N}_i^{\otimes l}(\varphi), \phi_l\right)
\geq 1 - o(l^0) \label{end_fid}
\end{align}
holds. Eqns. (\ref{end_fid}) and (\ref{ent_gen_rate}) show that
\begin{align*} 
 \underset{1 \leq i \leq N}{\min}I_c(\rho, \mathcal{N}_i) - \epsilon
\end{align*}
is an achievable rate.
\end{proof}
To conclude this section we compare the proof of Theorem \ref{ent_generation} given above with the one given in
Ref. \cite{bjelakovic09d}. The original achievability proof relies on the fact that good entanglement generation codes 
can be deduced from entanglement transmission codes working good on maximally mixed states on certain subspaces of the
input space of the channels. The passage to arbitrary states is done by a compound version of the 
so-called BSST-Lemma \cite{bennett02}. Indeed, one of the results from Ref. \cite{bjelakovic09d} is that the entanglement 
transmission capacity $\mathcal{Q}(\mathfrak{I})$ equals $E(\mathfrak{I})$ for every compound channel 
$\mathfrak{I}$.\newline
The proof given above follows a more direct route by taking advantage of a direct correspondence between entanglement 
distillation from bipartite states and entanglement generation over quantum channels, which is very close even in the 
compound setting. In this way, we have demonstrated, that quantum state merging provides a genuine approach to problems 
of entanglement generation over quantum channels even in the compound setting. 
\end{subsection}
\end{section}
\begin{section}{\label{sect:conclusion}Conclusion}
In this work, we have extended the concept of quantum state merging to the case, where the users are partially ignorant of
the parameters which describe the state they keep. We have determined the optimal entanglement cost of state merging in 
this setting, and found out that, in principle, a merging process is possible with the worst case merging cost in the set 
representing this uncertainty. We also derived a lower bound on the classical cost for merging with state uncertainty, 
based on an elementary proof of the corresponding result for single states. Whether or not this bound is achievable in 
general, is left as an open question. In particular, we have shown, that the class of protocols (called ``$L$-mergings'' 
in this work), which contains protocols optimal for the quantum as well the classical part of the state merging problem 
in case of perfectly known states is 
suboptimal in its classical costs for situations with state uncertainty. However, in some special cases, protocols which 
are minor variations of the $L$-merging concept achieve this bound. \\
Despite this, the protocol preserved its good reputation as a communication primitive regarding the quantum performance. 
We were able, to apply our results to prove corresponding assertions in other communication settings as entanglement 
distillation under state uncertainty as well as entanglement generation under channel uncertainty. To apply these results
in more complicated situations as multiuser settings (e.g. entanglement generation over quantum multiple access channels)
is an interesting topic for further research activities.
\end{section}

\section*{Acknowledgments}
We wish to thank Prof. K.-E. Hellwig, and J. N\"{o}tzel for their encouragement and many stimulating discussions. 
We also thank the Associate Editor for his/her comments on LOCC definitions which motivated us to include 
the appendix to this paper.\newline
The work of I.B. and H.B. is supported by the DFG via grant BO 1734/20-1 and by the BMBF via grant 01BQ1050.  

 \section{Appendix: LOCC Channels}
  In this section, we give a short account to the class of one-way LOCC channels which we use in our considerations.
 For further information, we recommend the survey article by Keyl \cite{keyl02} (and references therein). A more 
 recent general treatment can be found in Ref. \cite{chitambar12}. \newline 
 Crucial for the definition of LOCC channels is the concept of an instrument. Instruments (or operation 
 valued measures\cite{davies70}) were introduced to model the situation, where a measurement is made, and not 
 only the measurement results but also the state transformations according to the measurement values are taken into
 account. 
 To each measurement result $i$, there is assigned a positive trace non-increasing cp map 
 $\mathcal{I}_i$ which transforms the input state. In this paper, we restrict ourselves to finite sets of 
 possible measurement results.
  \begin{definition}
  A (finite) \emph{instrument} $\mathcal{A}$ is a map
  \begin{align*}
   \mathcal{A}: I & \rightarrow \mathcal{C}^{\downarrow}(\hr,\kr) \\
   i& \mapsto  \mathcal{A}_i &(i \in I)
  \end{align*}
  with a finite index set $I$ and Hilbert spaces $\hr$, $\kr$, such that $\sum_{i \in I} \mathcal{A}_i$ is 
  trace preserving. The instrument $\mathcal{A}$ is completely determined by the family $\{\mathcal{A}_i\}
  _{i\in I}$. We will sometimes write $\mathcal{A} = \{\mathcal{A}_i\}_{i\in I}$ to denote the instrument $\mathcal{A}$.
  \end{definition}
  For bipartite systems, an instrument at, say, $A$'s (the sender's) site can be combined with a
  parameter-dependent channel use, which is defined by a function 
  \begin{align*}
   \mathcal{B}: I &\rightarrow \mathcal{C}(\hr_B,\kr_B) \\
                i &\mapsto \mathcal{B}_i &(i \in I),
  \end{align*}
  i.e. each $\mathcal{B}_i$ is a completely positive and trace preserving map. A one-way LOCC channel is then 
  defined as a combination of an instrument and a parameter-dependent channel. This leads to the following 
  definition.
  \begin{definition} \label{def_app_locc_one_def}
  A channel $\mathcal{N} \in \mathcal{C}(\hr_{AB}, \kr_{AB})$ is called $A \rightarrow B$ 
  \emph{one-way LOCC channel}, if it takes the form 
  \begin{align}
   \mathcal{N}(\rho) = \sum_{i \in I} \mathcal{A}_i \otimes \mathcal{B}_i(\rho) &&(\rho \in \st(\hr_{AB})), 
   \label{app_locc_one_def}
  \end{align}
  where $\mathcal{A} = \{\mathcal{A}_i\}_{i \in I}$, $\mathcal{A}_i \in \mathcal{C}^{\downarrow}(\hr_A,\kr_A)$, 
  is an instrument and $\{\mathcal{B}_i\}_{i \in I}$ is a parameter-dependent channel.
 \end{definition}
 A one-way LOCC can also again be considered as a ``one-way local'' instrument\cite{chitambar12} with members 
 $\{\mathcal{A}_i \otimes \mathcal{B}_i\}_{i \in I}$. There is a convenient way of handling one-way LOCCs. 
 One can equivalently write the instrument $\mathcal{A}$ used on $A$'s site in channel form
 \begin{align*}
  \mathcal{A}(\rho) = \sum_{i \in I} \mathcal{A}_i(\rho) \otimes \ket{e_i}\bra{e_i} &&(\rho \in \st(\hr_A))
 \end{align*}
 with an orthonormal basis $\{e_i\}_{i \in I} \subset \cc^{|I|}$. If the basis is assigned to a system on $B$'s 
 site (which models a classical communication and coherent storage of the measurement results at the receiver's 
 system), the parameter-dependent channel can be written in the form
 \begin{align*}
  \mathcal{B}(\rho) := \sum_{i\in I} \ket{e_i}\bra{e_i} \otimes \mathcal{B}_i(\rho) &&(\rho \in \st(\hr_B))
 \end{align*}
 (this map may not not be trace-preserving). 
 Then we have for $\rho \in \st(\hr_{AB})$
 \begin{align*}
  \mathcal{N}(\rho) &= (id_{\kr_A} \otimes \mathcal{B}) \circ (\mathcal{A}\otimes id_{\hr_B})(\rho) \\
                    &= \sum_{j,i \in I} 
                          \mathcal{A}_i \otimes \mathcal{B}_j(\rho) \otimes \ket{e_i}\bra{e_i}\ket{e_j}\bra{e_j}  
                         \\
                    &= \sum_{i \in I} \mathcal{A}_i \otimes \mathcal{B}_i(\rho) \otimes \ket{e_i}\bra{e_i},
                          \end{align*}
where the second line includes a permutation of the tensor factors. Tracing out the classical information exchanged
within the application of the map (i.e. the system with space $\cc^{|I|}$) leads back to the form given in Eq. 
(\ref{app_locc_one_def}). The more general class of two-way LOCC channels exhibits a more intricate definition  
for which we refer to Refs. \cite{keyl02}, \cite{chitambar12}.\newline
Moreover, Def. \ref{def_app_locc_one_def} should not be confused with the definition of the class of separable 
channels. A channel $\mathcal{M} \in \mathcal{C}(\hr_{AB},\kr_{AB})$ is called \emph{separable}, if it takes the form
 \begin{align}
  \mathcal{M}(\rho) = \sum_{i \in I} \mathcal{A}_i \otimes \mathcal{B}_i(\rho) &&(\rho \in \st(\hr_{AB}))
   \label{app_sep_def}
 \end{align}
  where $\mathcal{A}_i \in \mathcal{C}^{\downarrow}(\hr_A,\kr_A)$ and $\mathcal{B}_i \in 
  \mathcal{C}^{\downarrow}(\hr_B,\kr_B)$ for all $i \in I$. From eqns. (\ref{app_locc_one_def}) and 
  (\ref{app_sep_def}), the difference between the one-way LOCC and separable channels can be observed. While 
  separable channels allow general trace decreasing cp maps for both parties, the receiver party is 
  restricted to usage of trace preserving cp maps (i.e. channels) in the one-way LOCC class of channels.

\nocite{*}

\end{document}